\documentclass[12pt,english]{amsart}

\usepackage[T1]{fontenc}
\usepackage[latin9]{inputenc}
\usepackage{geometry}
\usepackage{caption}
\usepackage{verbatim}
\usepackage{amsthm}
\usepackage{amstext}
\usepackage{amssymb}
\usepackage{setspace}
\usepackage[normalem]{ulem}
\usepackage{verbatim}
\usepackage{sgame}
\usepackage{xcolor}

\makeatletter
\numberwithin{equation}{section}
\numberwithin{figure}{section}

\usepackage{tikz}

\newcommand{\bs}{\boldsymbol}

\newcommand{\ud}{\underline\delta}

\newtheorem*{lemma*}{Lemma 3'}
\usepackage{color}

\theoremstyle{plain}
\newtheorem{thm}{\protect\theoremname}
\newtheorem{lemma}{Lemma}
\newtheorem{example}{Example}

  \theoremstyle{definition}
  \newtheorem{defn}{\protect\definitionname}
\providecommand{\definitionname}{Definition}
\providecommand{\theoremname}{Theorem}
\newtheorem{corr}{Corollary}
\usepackage{amsthm,amsmath,amssymb}

\theoremstyle{plain}

\newtheorem{proposition}{Proposition}
\newtheorem{corollary}{Corollary}
\newtheorem{conjecture}{Conjecture}

\theoremstyle{remark}

\newcommand{\dd}{{{\rm{d}}}}

\newcommand{\dl}{\delta}

\newcommand{\al}{\alpha}
\newcommand{\la}{\lambda}

	\newcommand{\uy}{\underline{y}}
	\newcommand{\oy}{\overline{y}}

\DeclareMathOperator*{\co}{co}

\DeclareMathOperator*{\interior}{int}

\DeclareMathOperator*{\supp}{supp}

\usepackage[foot]{amsaddr}

\newcommand{\ssymbol}[1]{^{\@fnsymbol{#1}}}

\makeatother

\begin{document}
\onehalfspacing
\title{BLACKWELL EQUILIBRIA IN REPEATED GAMES}

\author{Costas Cavounidis$\ssymbol{1}$}
\address{$\ssymbol{1}$ University of Warwick, c.cavounidis@warwick.ac.uk}

\author{Sambuddha Ghosh$\ssymbol{2}$}
\address{$\ssymbol{2}$ Shanghai University of Finance and Economics, CUHK, sghosh@cuhk.edu.hk}

\author{Johannes H{\"o}rner$\ssymbol{3}$}
\address{$\ssymbol{3}$  CNRS (TSE), johannes.horner@tse-fr.eu}

\author{Eilon Solan$\ssymbol{4}$}
\address{$\ssymbol{4}$ Tel Aviv University, eilons@tauex.tau.ac.il}

\author{Satoru Takahashi$\ssymbol{5}$}
\address{$\ssymbol{5}$ National University of Singapore, ecsst@nus.edu.sg}
\date{\today}

\begin{abstract}
We apply Blackwell optimality to repeated games. 
An equilibrium whose strategy profile is sequentially rational for all high enough discount factors simultaneously is a Blackwell (subgame-perfect, perfect public, etc.)  equilibrium.
The bite of this  requirement depends on the monitoring structure. Under perfect monitoring, a ``folk'' theorem holds relative to an appropriate notion of minmax. Under imperfect public monitoring, absent a public randomization device, any perfect public equilibrium generically involves pure action profiles or stage-game Nash equilibria only. Under private conditionally independent monitoring, in a class of games that includes the prisoner's dilemma, the stage-game Nash equilibrium is played in every round.
\bigskip

\end{abstract}
\thanks{We wish to thank Gabriel Carroll, Aniruddha Dasgupta, Olivier Gossner, Bart Lipman,
Elliot Lipnowski, Francesco Nava, Jawwad Noor, Juan Ortner, Phil Reny, Takuo Sugaya, and Balazs Szentes,
as well as audiences at Boston University, the Econometric Society Summer Meetings, LSE, MIT, NYU, and the University of Warwick for useful comments and discussions.
Solan thanks the support of the Israel Science Foundation, Grants \#217/17 and \#211/22.
}
\maketitle
\bigskip{}


\newpage{}












\newpage

\section{Introduction}


By and large, the economic literature on repeated games has adopted discounting as the payoff criterion. It is technically convenient, and captures the idea that the distant future does not matter much for current decisions, which is certainly ``more realistic than its opposite,'' as a leading microeconomics textbook puts it.\footnote{See Mas-Colell, Whinston, and Green (1995), p.734.} Yet, at least in the context of repeated games, it has two consequences that are often viewed as undesirable. First, no action profile can typically be ruled out; in this sense little is said about
behavior.\footnote{See the discussion in Aumann and Maschler (1995), p.139.} Second, predictions depend on common knowledge of the exact
discount factor, undoubtedly a strong assumption.\footnote{This second issue has  led many game theorists (Aumann and Maschler, among many others) to favor undiscounted payoff criteria. This is throwing the baby out with the bathwater, as forsaking impatience reinstates the unrealistic ``opposite'' mentioned above.}

In this paper, we study \textit{Blackwell equilibria}, that is, equilibria whose strategy profiles are optimal for all high discount factors simultaneously. Hence, they preserve the property that time isn't free, and that every round matters for the player's payoff, yet, by definition, they cannot depend on the exact value of the discount factor. This is a payoff criterion, not a solution concept. For the latter, we adopt what is commonly used depending on the environment: subgame-perfect Nash equilibrium, perfect public equilibrium, sequential equilibrium, etc. With some abuse, we refer to the relevant notion of equilibrium under the Blackwell criterion as Blackwell equilibrium. The name ``Blackwell equilibrium'' derives from ``Blackwell optimality,'' the corresponding concept introduced for Markov decision processes by Blackwell (1962).

Robustness to the exact discount rate admits several interpretations. When the rate is thought of as arising from the random length of the actual interaction, there are many situations in which players are uncertain about exactly how long this interaction will take place, and this uncertainty might be sufficiently vague that modeling it explicitly seems futile.\footnote{As Aumann and Maschler (1995, p.133) put it, ``there is a limit to the amount of detail that can usefully be put into a model, or indeed that the players can absorb or take into account.''} The same applies when the discount rate pertains to the players' time preferences.\footnote{Admittedly, in that case, by revealed preference, a player ``knows'' his own discount factor. The case where each player knows his own discount factor, but not the others', is taken up in a separate short note available by request from the authors. } Uncertainty regarding future interest rates is both subjective and significant; and it has large, negative and persistent effects on the economy (see, \textit{e.g.}, Istrefi and Mouabbi, 2018).\footnote{As discussed below, our analysis is readily adjusted to the case in which the discount factor isn't constant over time, or deterministic.} When players in the game are convenient proxies for groups of agents (countries, political parties, firms, etc.), then Blackwell equilibria have the desirable feature that they are unanimously viewed as optimal by the constituents of each group, independent of exactly how patient each of them is, provided that they are all sufficiently patient. Finally, from the point of view of the analyst, they allow to predict, or explain, behavior that might apply to a variety of situations, which might differ in the details of the interaction length.

Our goal is to understand how this more restrictive payoff criterion affects the usual predictions about payoffs and action profiles in infinitely repeated games, under various monitoring structures. Our main result is that its (relative) bite is increased as monitoring ``worsens,'' so to speak. Loosely speaking, this is because robustness to discounting makes it difficult to enforce mixed actions. Yet, the role of mixed strategies becomes progressively more important as the information structure shifts from perfect monitoring, to imperfect public monitoring, and finally to imperfect private monitoring.

 In games of perfect monitoring, Blackwell (subgame-perfect) equilibria still span a large set of equilibrium payoffs, but not as large as under the (limit of) discounting criterion. Indeed,  the standard folk theorem (see Fudenberg and Maskin, 1986, hereafter FM) requires punishments that, depending on the stage game, might involve mixing by the punishing players. To make the players indifferent over the support of their actions, these players must be compensated in the continuation game, as a function of the action they have chosen. To achieve indifference, this compensation must be finely tuned to the discount factor. We show that there is no way around this difficulty. As a result, a new notion of minmax payoff must be introduced, capturing the fact that punishing players must be myopically indifferent across all actions within the support of their mixed action (but not necessarily over all actions available to them).

We show that this is the only adjustment that must be made to the ``standard'' statement of the folk theorem -- indeed, mixed actions play no other role in the usual proofs under perfect monitoring. The construction involves the same ingredients as the proof of  FM, and can be achieved using a variation of simple strategies (Abreu, 1988). Not too surprisingly, this folk theorem can be extended to imperfect public monitoring, in the special case in which monitoring satisfies product structure, individual full rank, and a public randomization device is available.\footnote{Admittedly, the product structure is very special, but it applies to important classes of games, such as games with one-sided imperfect monitoring, e.g. principal-agent games, and games with adverse selection and independent types.}

In general, under imperfect public monitoring, it is known that unpredictable behavior serves another purpose. Mixed actions enlarge the set of detectable deviations, and hence affect the sufficient conditions under which the folk theorem usually holds (Fudenberg, Levine, and Maskin, 1994, hereafter FLM). Yet, the impossibility of fine-tuning continuation payoffs in order to compensate players for mixing in a way that would be independent of the discount rate further restricts the action profiles that can be implemented. Absent a public randomization device, only stage-game Nash equilibria and pure action profiles can be played in a (perfect public) Blackwell equilibrium (generically, see Proposition \ref{prop:generic}). That is, the only mixed actions that can be played are stage-game Nash equilibria: it no longer suffices that players be myopically indifferent over the support of their mixed action. This is because, unless the action profile is a Nash equilibrium of the stage game, the continuation play must depend on the realized signal, which makes it impossible for players to be indifferent over multiple actions (as they generically induce distinct distributions over public signals), even if they are myopically indifferent over those.

A major difficulty in the analysis under imperfect monitoring is that such games are usually studied via recursive techniques involving the set of equilibrium payoffs (see Abreu, Pearce, and Stacchetti, 1990). Because payoffs of a given strategy profile, and  the ``self-generation'' operator itself, depend on the discount rate, standard results are not as helpful here, since optimality must hold for an entire range of discount rates simultaneously.\footnote{The alternative route in the literature involves review strategies, following Radner (1985). Unfortunately, review strategies don't specify behavior fully, making it difficult to ensure that behavior is independent of the discount rate.} Hence, our analysis must tackle directly the issue of the action profiles that can be enforced.

Finally, we show that behavior is further constrained once monitoring is private. In a  class of games that includes the prisoner's dilemma, when monitoring satisfies conditional independence, the only Blackwell equilibrium outcome consists in the repetition of the stage-game Nash equilibrium. This is because indifference between actions is known to play a further role under private monitoring. Under conditional independence, with pure strategies, a player cannot tell, even statistically, whether his opponent is supposed to ``punish or reward'' him; therefore, his opponent cannot be incentivized to select one or the other continuation strategy as a function of the signals he receives, unless he happens to be indifferent across those (see Matsushima, 1989). Hence, any non-trivial sequential equilibrium must involve indifferences (whether a player actually mixes or uses his private history to select one or the other continuation strategy). For the same reason as under public monitoring, such indifference is inconsistent with the robustness to the discount rate.

\subsubsection*{Related Literature}
The Blackwell optimality criterion has been introduced by Blackwell (1962) for finite Markov decision processes, as a way of characterizing optimality in the undiscounted case. Blackwell shows that optimal policies exist, and provides (a pair of) optimality equations to solve for those. More recently, this criterion has been applied to stochastic games, both in discrete time (Singh,  Hemachandra and Rao, 2013) and in continuous time (Singh, Hemachandra, 2016).  The focus of these papers is to provide conditions under which (Nash) equilibria exist under this payoff criterion. They do this for games in which a single player controls the transitions and the payoff of the non-controller is additive in the players' actions.
Indeed, existence is a non-trivial problem in the environments they consider. In repeated games, this is immediate, as the repetition of stage-game Nash equilibria is a Blackwell equilibrium. In contrast to the above papers, we are interested in characterizing the set of such equilibria under different monitoring structures.

The motivation of our paper is related to Gossner (2020). Gossner's goal is also to define equilibria that are robust to slight perturbations in the repeated game.  Gossner introduces incomplete penal codes as partial descriptions of equilibrium strategies, and studies to what extent such codes can be found, whose completion is allowed to depend on the fine details of the game. Because his class of perturbations includes not only the discount rate but also the payoff matrix itself, complete penal codes typically do not exist.

Slightly less related are some papers that focus on special classes of strategies.  Kalai, Samet and Stanford (1988) study reactive equilibria, which  are equilibria in which at least one player conditions his actions only on his own opponent's action and not on his own past actions. As they show, if a reactive strategy profile is robust to nearby discount factors, it must play the stage-game Nash equilibrium in the prisoner's dilemma. 


Section 2 illustrates our results with the prisoner's dilemma. Section 3, 4 and 5 state our general results for perfect, imperfect public, and imperfect private monitoring respectively.

 \section{An Example}
 
 To motivate our analysis and illustrate our results, it suffices to consider the normal-form game \textit{par excellence}, the prisoner's dilemma, which captures the trade-off between cooperation and competition, from economics to political science. To be concrete, let us adopt the parameters from Mailath and Samuelson (2006, Section 7.2), namely:
 
\begin{figure}[h]
\centering
	\begin{game}{2}{2}
		& $C$    & $D$ \\
		$C$ & $2,2$ & $-1,3$ \\
		$D$ & $3,-1$ & $0,0$ 
	\end{game}
\end{figure}
Under perfect monitoring, the infinitely repeated version of the game lacks subtlety: any strictly individually rational payoff vector can be supported as a strict equilibrium, involving a rotation among pure action profiles with the requisite frequencies, provided only that players are patient enough to ensure that the threat of Nash reversion counters the potential immediate gain from defecting. Given that incentives are strict, uniformly across histories, the specific discount rate is irrelevant, as long as it is common knowledge that it is below some critical threshold.

Cooperation gets more complicated to sustain when actions are not perfectly monitored. Again, we follow Mailath and Samuelson and consider the case in which the public signal is either $\oy$ or $\uy$, with the probability of observing the high signal $\oy$ given the action profile $a \in A:=\{C,D\}\times\{C,D\}$ being given by
$$
\pi(\oy \mid a)=
\begin{cases}
p \mbox{ if } a= CC,\\
q \mbox{ if }a =CD,DC\\
r \mbox{ otherwise}.
\end{cases}
$$
Specifically, we choose $p=9/10$, $q=8/10$ and $r=1/10$. Think of the two players as Cournot duopolists: the more they cooperate by restricting output, the more likely it is that the publicly observed market price is high.\footnote{See Mailath and Samuelson, Figure 7.2.1, for the \textit{ex post} payoff matrix, defined in terms of the signals only, giving rise to the payoffs considered here.}

What are the equilibrium payoffs that can be supported by strict equilibria, provided only that players are  patient enough? The characterization of Fudenberg and Levine (1994) provides the  tool necessary to identify this set $\lim_{\delta\rightarrow 1}E^p_\delta$.\footnote{That is, provided attention is restricted to  public perfect equilibria, as we do here.} As they show, it suffices to consider the parameterized program, for $\lambda \in \mathbf{R}^2$, $\| \lambda \|=1$, given by
\[
k(\lambda)=\sup_{\substack{ a \in A,\ v\in \mathbf{R}\\\{x(y) \in \mathbf{R}^2\}_{y=\uy,\oy} }} \lambda \cdot v,
\]
such that (``$v$ is a pure-strategy Nash payoff of the game augmented by $x$'')\footnote{The set of Nash equilibria of the game where players in $N$ pick action profiles in $A$, resulting in payoffs according the utility functions $\hat{u}$, is denoted here by $\textrm{NE} (N, A, \hat{u} )$.}
\[
v = \hat{u}(a^*) \mbox{ for some } a^* \in \textrm{NE} (N, A, \hat{u} ), \textrm{ where } \hat{u}(a) := u(a)+\sum_y \pi(y \mid a)x(y);
\]
and $\lambda \cdot x(y) \le 0,\ y=\uy,\oy$  (``weighted budget-balance'').

Let  $\mathcal{H}:=\cap_\la \{ v \in \mathbf{R}^2 \mid \lambda \cdot v \le k(\la)\}$. If $ \mathcal{H}$ has non-empty interior, then $\lim_{\delta\rightarrow 1}E^p_\delta= \mathcal{H}$. While there are infinitely many ``directions'' $\lambda$ to consider, it isn't difficult to solve for $\mathcal{H}$ in practice: after all, for each of the four action profiles, this is a linear programme in $x$.\footnote{In fact, as shown in H\"{o}rner, Takahashi and Vieille (1994), by considering the dual program, only finitely many programmes must be solved, given that there are finitely many pure action profiles.} In this particular instance, 
\[
\lim_{\delta\rightarrow 1} E^p_\delta =\mathrm{co}\ \{(1/4,7/4),(7/4,1/4),(0,0)\}, 
\]
a set (given by the blue triangle $OPQ$ in Figure \ref{fig1}) whose upper edge coincides with the line $v_1+v_2 = 2$.\footnote{All formal calculations are in the supplementary appendix.}

\begin{figure}
\begin{center}

\begin{tikzpicture}[scale=1.5]

\coordinate (O) at (0,0);
\coordinate (A) at (-1,3);
\coordinate (B) at (2,2);
\coordinate (C) at (3,-1);

\coordinate (P) at (0.25, 1.75);
\coordinate (Q) at (1.75, 0.25);

\coordinate (P') at (0.15, 2);
\coordinate (Q') at (2, 0.15);
\coordinate (R') at (1.2, 1.2);

\draw[->] (-1.5,0) -- (3.5,0) node[below] {};
\draw[->] (0,-1) -- (0,3.5) node[left] {};

\draw[dashed] (A) -- (C);
\draw[dashed] (O) -- (B);

\draw[thick] (O) -- (A) -- (B) -- (C) -- cycle;

\draw[thick,blue] (O) -- (P) -- (Q) -- cycle;

\draw[thick,red] (O) -- (P') -- (R') -- (Q') -- cycle;

\fill (O) circle (1pt) node[below left] {$O = (0,0)$};
\fill (A) circle (1pt) node[left] {$(-1,3)$};
\fill (B) circle (1pt) node[above right] {$(2,2)$};
\fill (C) circle (1pt) node[right] {$(3,-1)$};
\fill (P) circle (1.5pt) node[below right] {$P$};
\fill (Q) circle (1.5pt) node[above left] {$Q$};
\fill (P') circle (1.5pt) node[above right] {$P'$};
\fill (R') circle (1.5pt) node[above left] {$R'$};
\fill (Q') circle (1.5pt) node[right] {$Q'$};

\end{tikzpicture}

\end{center}
\caption{Pure-strategy (blue) and mixed-strategy (red) limit payoff sets (Not to scale).}
\label{fig1}
\end{figure}
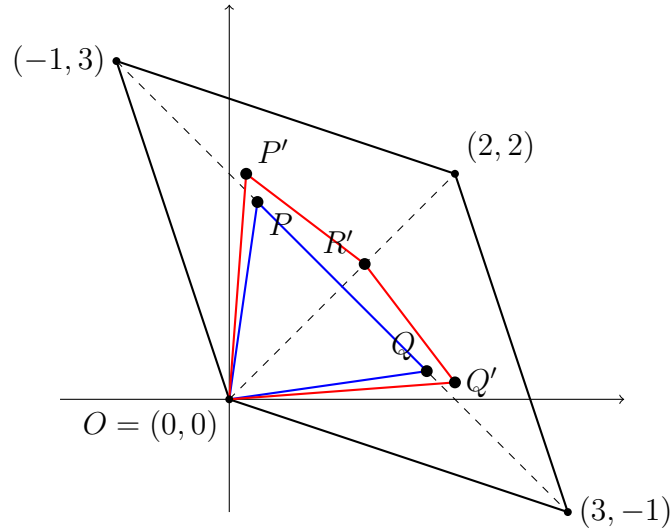

Note that payoffs are bounded away from the efficient payoff vector $(1,1)$, no matter how patient the  players. Because there are only two signals, deviations cannot be attributed to a particular player (formally, the monitoring structure fails \textit{pairwise full rank}). As a result, statistical evidence of defection must lead to surplus destruction, rather than to surplus shift  (from the ``statistically guilty'' to the ``statistically virtuous'' player). Since such evidence is bound to occur, so does such destruction.

Cooperation might be limited here, but at least it involves strict incentives at all times: as a result, here again, for any payoff vector (strictly) within this set to be an equilibrium payoff, it suffices that players commonly know that the discount rate be lower than some critical threshold: more information isn't called for.

Could our conspirators do better by playing mixed strategies? To answer this question, here again, it suffices to apply Fudenberg and Levine's algorithm, replacing $a \in A$ with $\al \in \Delta(A_1)\times \Delta(A_2)$ wherever $a$ occurs above. The programme becomes truly non-linear, as attention cannot be restricted to finitely many action profiles. Yet, in this instance, it isn't difficult to solve for the limit set of equilibrium payoff  $\lim_{\delta\rightarrow 1}E _\delta$ involving mixed strategies as well. Formally,
\begin{equation*}
 \resizebox{1\hsize}{!}{$  \lim_{\delta\rightarrow 1} 
\ E _\delta = \mathrm{co}  \left\{
(0,0),
\left(\frac{14-2\sqrt{15}}{36-3\sqrt{15}},\frac{98-14\sqrt{15}}{36-3\sqrt{15}}\right),
 \left(\frac{7-\sqrt{15}}{3}, \frac{7-\sqrt{15}}{3}\right), 
\left(\frac{98-14\sqrt{15}}{36-3\sqrt{15}},\frac{14-2\sqrt{15}}{36-3\sqrt{15}}\right)
\right\}$},
\end{equation*}
but  Figure  \ref{fig1} is certainly more eloquent --- see the red quadrilateral $OP'R'Q'$, which is not drawn to scale for the sake of visual clarity. 

The improvement is real even if it isn't blindingly obvious: using mixed actions increases efficiency. To see why, note that defecting might be bad for payoffs, but it drastically improves the statistical power of the signals: when both players defect, the low signal becomes very likely, whereas all other action profiles induce very similar signal distributions. Hence, randomizing allows players to improve on their ability to detect defections, at the cost of defecting themselves!

However, such sophisticated play requires players to be willing to randomize. Because defecting is dominant in the stage game, continuation play must be finely tuned to compensate them for cooperating. Intertemporal preferences matter, and so does the knowledge of the discount rate. One of our main results is to show that, no matter how patient players are commonly known to be, it isn't possible to implement such play, \textit{unless} the discount factor is precisely known, in a sense made precise below.

What if monitoring isn't public? Randomization is useful in that case as well. Indeed, the belief-based construction of Sekiguchi (1997), or the belief-free ones of Piccione (2002) or Ely and V\"{a}lim\"{a}ki (2002) (see Mailath and Samuelson, 14.1.1 for details in the context of our specific prisoner's dilemma example) rely on it. To be sure, it is possible that this randomization could be purified by using a player's private history as a purification device. Nonetheless, indifference over pure strategies is critical, as shown by Matsushima (1991), at least when monitoring has a product structure. Our last main result leverages his result and implies that, in our example, with product structure, absent precise knowledge of discounting, the only equilibrium play involves perpetual defection.

While this example provides a good illustration of our main results, it is too simple to capture all of them. In particular, mixing -- which as we show, is severely curtailed, but not necessarily precluded under imprecise discounting -- can also play a role under perfect monitoring, albeit a limited one, to the extent that punishing players might require punishers to randomize. Accordingly, this is where our formal analysis starts.

\section{Blackwell Equilibria Under Perfect Monitoring}

This section studies Blackwell equilibria  under perfect monitoring. First, we derive necessary conditions that such equilibria must satisfy. This leads to a modified notion of minmax payoff, and to a folk theorem relative to this notion.

\subsection{Notation and Definitions}

The set of players is $I=\{1,\ldots,n\}$. Player $i$'s finite set of actions is $A_{i}$, and $A:=\prod_{i\in I}A_{i}$ is the set of all pure action profiles. A mixed action of $i$ is $\alpha_{i}\in\Delta  A_{i}$, where $\Delta E$ is the set of all probability distributions on a set $E$; the set of (independent) mixed-action profiles is
$\mathcal {A} := \prod_i \Delta A_i$. Player $i$'s reward function is a map $g_{i}:A\rightarrow\mathbb{R}$, whose domain is extended to $\Delta A$ in the usual way, and $g:=(g_i)_{i=1}^n$. The set of feasible payoffs is $F:=\co(g(A))$, where $\co$ denotes the convex hull. Denote this normal-form game by $G=\left<I;A,g\right>$.

The stage game $G$ is played at each $t\in\mathbb{Z}_{+}$. Denoting by $a^{(t)}\in A$  the action profile chosen in each round $t$, the history at the end of round $t\in \mathbb{Z}_+$ is $h^{t}=( a^{(1)},\ldots, a^{(t)})\in A^{t}=:H^{t}$, with $h^{0}$   the empty history and $H:=\cup_{t=0}^\infty H^t$ the set of all histories. An outcome $h^\infty$ is an infinite sequence  $(a^{(t)})_{t=1}^\infty$. 
 Given discount factor $\delta_i \in[0,1)$, player $i$'s (average discounted)  payoff is defined as
\begin{equation}
U_{i} \left( h^\infty ,\delta_i\right)
:=(1-\delta_{i})
{\displaystyle \sum_{t=1}^{\infty}} {\delta_{i}}^{t-1}g_{i}(a^{(t)}).
\label{eq:discounted-u}
\end{equation}
This defines the repeated game $G^{\infty}({\bs{\delta}})$, where the vector $\bs{\delta}=(\delta_{1},\dotsc,\delta_{n})$ is referred to as the discount factor vector; $G^{\infty}(\delta)$ is the special case with common discount factor $\delta$.\footnote{Bold symbols are used for only those vectors whose scalar counterparts are also used, \textit{e.g.}, $\bs{\delta}$.}

A pure strategy of player $i$ is a function $s_i:H\rightarrow A_{i}$;
a behavioral strategy is a function $\sigma_i:H\rightarrow \Delta A_{i}$.
The set of (behavioral) strategies for $i$ is denoted $\Sigma_i$, and the set of (behavioral)
strategy profiles is denoted $\Sigma$. Player $i$'s expected payoff (or payoff, for short) given $\sigma$, $U_{i}(\sigma, \delta_i)$,  is defined the usual way. The payoff vector is denoted $ U(\sigma, \bs\delta)$.%

Unless mentioned otherwise, no public randomization device (PRD) is assumed.
%

\subsection{Blackwell Equilibrium}

Given that monitoring is perfect, the natural solution concept is subgame-perfect Nash equilibrium (SPNE).
\begin{defn}\label{defi} A strategy profile $\sigma \in \Sigma$ is a \textbf{Blackwell SPNE above} $\bs{\underline{\delta}}$ if there exists $\underline{\delta}\in[0,1)$ such that $\sigma$ is an SPNE of $G^{\infty} (\bs \delta)$ at any $\bs\delta \geq\underline{\delta}\cdot (1,\ldots,1)$.\\
A vector $v\in  \mathbb{R}^n$ is a \textbf{Blackwell SPNE payoff at} $\bs\delta$ if there exists a Blackwell SPNE $\sigma$ above some $\underline{\delta}$, with
$\bs\delta\geq\underline{\delta}\cdot (1,\ldots,1)$, such that $v=U(\sigma;\bs\delta)$.
\end{defn}
We say that a strategy profile is a Blackwell SPNE if it is a Blackwell SPNE above some $\underline{\delta}<1$; similarly, we refer to a Blackwell SPNE payoff. 
A more general definition would allow discounting to vary with time (or even with the history). That is, one might consider an evaluation $(\delta_t)_{t=1}^\infty$ (a probability distribution over positive integers), where the weight of any round $t$ is given by $\delta_t$ (see Renault, 2014). This might be particularly relevant in settings where discounting captures the uncertainty in the length of the interaction. The choice adopted here is made primarily for simplicity. Plainly, enlarging the set of discount sequences that the equilibrium must survive further restricts the set of equilibria; yet, our equilibrium constructions do not rely on the fact that the sequences we focus on are constant (or, for that matter, deterministic). Similarly, we could weaken the criterion without changing any result by requiring optimality to hold only for constant vectors $\bs \delta=(\dl,\ldots,\dl)$, if this is more appropriate in some context.

\subsection{A Necessary Condition\label{subsec:BE-necessary-PM}}

Subgame-perfection puts little restriction on action profiles specified by an equilibrium, as long as the feasible and individually rational payoff set has non-empty interior. Matters are different under the Blackwell criterion.

 The set $\mathcal{A}^\mathrm{MI}$ are those mixed action profiles $\alpha$ such that each player gets the same  reward  from each action in the support of $\alpha_i$, given $\alpha_{-i}$.
 This property is called \textbf{Myopic Indifference} (MI).
Formally,
\begin{equation}
\mathcal{A}^\mathrm{MI}
:=\left\{ \alpha\in\mathcal A\,\Big|\ g_{i}(a_{i},\alpha_{-i})
=g_{i}(\alpha)\ \forall
i\in I,\forall a_{i}\in\supp(\alpha_{i})\right\}.
\label{eq:X}
\end{equation}
To put it differently, $\alpha$ is in $\mathcal{A}^\mathrm{MI}$ if, and only if, it is a Nash equilibrium of the stage game   $\left<I;(\supp(\alpha_i))_{i\in I},(g_{i})_{i \in I}\right>$. 
Notice that $\mathcal{A}^\mathrm{MI}$ contains all pure-action profiles (with $\supp(a_i)=\{a_i\}$) and Nash equilibria (with $\supp(\alpha_i)=A_i$).

The motivation for the definition of $\mathcal{A}^\mathrm{MI}$ derives from the following result.
\begin{proposition}\label{prop:X}
If $\sigma$ is a Blackwell SPNE, then $\sigma(h)\in\mathcal{A}^{\mathrm{MI}}$ for any history $h\in H$.
\end{proposition}

\begin{proof} If $\sigma$ is a Blackwell SPNE, it is an SPNE at all $\delta$ in an open interval $\mathcal{O}\subset(0,1)$. Fix any history $h^{t-1}$, player $i\in I$, and actions $a_{i},a_{i}^{\prime}\in \supp(\sigma_{i}(h^{t-1}))$. Let player $i$'s expected reward in round $\tau>t$ under the continuation strategy $\sigma|_{h^{t-1}}$ following the action $a_{i}$ (resp., $a_{i}^{\prime}$) at $t$ be $g_{i}^{(\tau)}$ (resp., $g_{i}^{\prime(\tau)}$).
Since player $i$ mixes over $a_i$ and $a'_i$, they yield the same payoff for any $\delta \in \mathcal O$:
\[
g_{i}(a_{i},\sigma_{-i}(h^{t-1}))+\sum_{\tau>t}\delta_{i}^{\tau-t}g_{i}^{(\tau)}
=g_{i}(a_{i}^{\prime},\sigma_{-i}(h^{t-1}))
    +\sum_{\tau>t}\delta_{i}^{\tau-t}g_{i}^{\prime(\tau)}\,~\forall\delta_{i}\in\mathcal{O},
\]
and hence
\begin{equation}
f(\delta_{i})\,
:=g_{i}(a_{i},\sigma_{-i}(h^{t-1})) -  g_{i}(a_{i}^{\prime},\sigma_{-i}(h^{t-1}))
    +  \sum_{\tau>t}\delta_{i}^{\tau-t}(g_{i}^{(\tau)} - g_{i}^{\prime(\tau)})
       =0\,~\forall\delta_{i}\in\mathcal{O}.\label{eq:power-series-eq}
\end{equation}
The Identity/Uniqueness Theorem (see Ahlfors (1953), p.127) implies that if the set of zeros of an analytic function has an accumulation point in its domain, then it is identically zero; since (\ref{eq:power-series-eq}) holds for an open interval of $\delta_i$, it follows that $f$ is identically zero in $(-1,1)$; in particular, setting $\delta_i=0$ gives:
\[
g_{i}(a_{i},\sigma_{-i}(h^{t-1})) =  g_{i}(a_{i}^{\prime},\sigma_{-i}(h^{t-1})).
\]
Thus, both $a_i$ and $a^\prime_i$ yield the same  reward; hence, $\sigma(h^{t-1})\in \mathcal{A}^\mathrm{MI}$. This shows that myopically indifferent action profiles must be played after any history.
{} \end{proof}
The strength of subgame-perfection is not needed for the conclusion: if attention is restricted to histories on path, the same holds for Nash equilibria. 

Standard constructions in the literature rely on action profiles that are not in $\mathcal{A}^\mathrm{MI}$. More specifically, such action profiles enter in the definition of the minmax payoff, namely
\begin{equation}
\underline{v}_i
:=\;\min_{\alpha_{-i}\in \prod_{j \neq i}(\Delta A_{j})}
\max_{a_{i}\in A_i}\; g_{i}(a_{i},\alpha_{-i}).~~
%
\end{equation}
To hold a player to this payoff, the other players may have to randomize over actions over which they are not myopically indifferent. 
This is typically the case for SPNE in FM, which are therefore not Blackwell SPNE. Since any state-game Nash equilibrium satisfies myopic indifference, SPNE in Friedman (1971) are Blackwell. 

In view of Proposition \ref{prop:X},
we introduce the following notion of $\textrm{MI}$-minmax payoff:
\begin{equation}
\underline{v}_{i}^{\mathrm{MI}} :=
\min_{\alpha\in\mathcal{A}^\mathrm{MI}}\;
\max_{a_{i}\in A_{i}}g_{i}(a_{i},\alpha_{-i}).\label{ri}
\end{equation}
Every pure action profile is in $\mathcal{A}^\mathrm{MI}$. It follows that
\[
\underline{v}_{i}^{\mathrm{MI}} \le \underline{v}_{i}^{\mathrm{pure}}:=\min_{ a_{-i}\in A_{-i}}\max_{a_{i}\in A_i}\; g_{i}(a_{i},\alpha_{-i}).
\]
Similarly, every Nash equilibrium of the stage game is included in $\mathcal{A}^\mathrm{MI}$. Hence, it also holds that $\underline{v}_{i}^{\mathrm{MI}} \le \underline{v}_{i}^{\mathrm{NE}}$, where $\underline{v}_{i}^{\mathrm{NE}}$ is player $i$'s lowest stage-game Nash equilibrium payoff.

The following example shows that the inequalities $\underline{v}_i \le \underline{v}_{i}^{\mathrm{MI}}  \le \underline{v}_{i}^{\mathrm{pure}} $ can be strict.

\begin{example}
\label{Ex} Consider the  payoff matrix given by Figure 1.

\begin{table}[h]
\centering
\begin{game}{2}{3}[Player 1][Player 2]
        &  $L$                 & $M$  & $R$   \\
     $T$  &~$\vphantom{\big[}(1,0)~$      & $~(0,0)~$ & $~(0,3)~$ \\
     $B$ &~$\vphantom{\big[}(0,0)$~ & ~$(3,0)$ ~& ~$(1,1)~$\\
\end{game}

\vspace{3mm}
\caption{\hspace{3mm} Figure 1: The game in Example~\ref{Ex}.}
\end{table}

To minmax player 1, player 2 must play $(\frac12L+\frac12R)$. However, as
$R$ is a dominant action for player 2, he cannot be myopically indifferent between $L$ and $R$. The 
$\textrm{MI}$-minmax of player 1 obtains when player 2 plays $(\frac34L+\frac14M)$, which is not as harsh a punishment, but
still worse for player 1 than pure minmaxing via (say) $R$, which is also player 2's action under Nash reversion. Hence, it holds that $\underline{v}_1 =\frac12 < \underline{v}_{1}^{\mathrm{MI}}=\frac34  < \underline{v}_{1}^{\mathrm{pure}}=\underline{v}_{1}^{\mathrm{NE}}=1 $.
\end{example}

Proposition \ref{prop:X} immediately implies the following.
\begin{corollary}
Every Blackwell equilibrium payoff $v$ satisfies $v_{i}\geq \underline v_{i}^{\mathrm{MI}}$, for all $i\in I$.
\end{corollary}



\subsection{A ``Folk'' Theorem\label{subsec:folk-PM}}
Under discounting, and subject to a mild dimensionality condition, FM establish a folk theorem for subgame-perfect Nash equilibrium:
given any $v\in F$ such that for all $i$, $v_i>\underline v_i$, there
exists $\underline{\delta}\in [0,1)$ such that, for any $\delta \in (\underline{\delta},1)$ there is a subgame-perfect Nash equilibrium $\sigma$ of $G(\delta)$ with payoff $U(\sigma,\delta)=v$. Clearly, the same cannot hold under the Blackwell criterion, given Proposition \ref{prop:X}.

Define
\begin{equation}
F^{\mathrm{MI}}:=\{v\in F \mid v_{i}>\underline v_{i}^{\mathrm{MI}},
\ \ \ \forall i\in I\}.\label{FRES}
\end{equation}
Whenever $F^{\mathrm{MI}}$ is full-dimensional, Proposition \ref{prop:X} implies that (the closure of) $F^{\mathrm{MI}}$ is an upper bound
 on the set of Blackwell SPNE payoff vectors. The following
theorem shows that this upper bound is tight in general.\footnote{Note that the statement of Theorem 1 refers to the equilibrium payoff vector evaluated at a common discount rate (yet, the strategy profile must be optimal for all possibly distinct discount factors high enough). This is because, as is well known (see Lehrer and Pauzner, 1999), the set of feasible payoffs evaluated at different discount factors can be larger than the convex hull of the stage-game payoffs, a topic that is orthogonal to our purpose.}



\begin{thm}\label{thm:perfect} Suppose that the dimension of $ F^{\mathrm{MI}}$ is $n$.\footnote{Unlike in FM, the full dimensionality assumption in Theorem \ref{thm:perfect} cannot be dropped for the case $n=2$ in general. Presumably, the general case (without interiority assumption) can be dealt with by adapting the notion of effective minmax (Wen, 1994) to account for the constraint that action profiles must be in ${\mathcal A}^{\textrm{MI}}$, along the lines of Fudenberg, Levine, and Takahashi (2007). } 
For any  $v\in F^{\mathrm{MI}}$, there exists $\underline \delta <1$ such that for all $\delta \in (\underline \delta,1)$, $v$ is a Blackwell SPNE payoff at $\delta$.
\end{thm}





The  proof,
which appears in Appendix~A,
 follows FM (see also Abreu, 1988) in having stick-and-carrot punishment regimes, one for each player. Any unilateral deviation from the prescribed strategies leads to a ``stick-and-carrot regime.'' In FM, the stick phase involves minmaxing, during which the player who deviated is held to his minmax payoff. Play then moves to the carrot phase, in which all players earn strictly more than their minmax payoff.

In the case in which the target payoff $v$ is achieved by a pure action profile, our construction is a straightforward adaptation of this construction. During the stick phase of player $i$, replace standard minmaxing with an action profile $\alpha^i\in\arg\min_{\alpha\in\mathcal{A}^\mathrm{MI}}\;
\max_{a_{i}\in A_{i}}g_{i}(a_{i},\alpha_{-i})$. For each player $j\neq i$, actions within the support $\alpha_j^i$ yield $j$ the same reward, and the selected one within it is subsequently ignored, whereas actions outside of the support are deterred as in FM.

What if $ v$ is not the payoff of a pure action profile? Lacking a PRD, we follow Dasgupta and Ghosh (2021) to construct action profile paths that deliver the target payoff while also keeping continuation payoffs near the target. We then show that if continuation payoffs given a pure action path and a certain discount remain bounded above and below, those same bounds apply at larger discount factors.\footnote{This is reminiscent of the Arrow-Levhari (1969) stopping theorem: if the value of a discardable security is weakly positive when evaluated at a certain discount (given optimal discarding), it is weakly positive at greater discounts. The same intuition applies here: any short-term setbacks and gains are smoothed out at higher discount factors. Therefore, continuation payoffs remain nearby at all higher discount factors, and hence deviations are deterred by the same punishments.} This ensures that continuation payoffs are similar enough to be enforced by the same punishments for all high enough discounts.

\section{Imperfect Public Monitoring}\label{sec:IM}

\subsection{Generic Games}\label{section:generic}

This section turns to  imperfect public monitoring, starting with the same finite set of players $I=\{1,2,...,n\}$ and finite sets of actions $A_i$, $i \in I$, with reward function $g_i:A \rightarrow \mathbb{R}$. 
A monitoring structure $(Y, \pi) $ is a finite set of signals $Y$ and a function $\pi\colon A\to \Delta Y$
mapping action profiles $a \in A$ into distributions over $Y$,
indicating the probability that each signal $y \in Y$ is publicly observed.
Let $ G= \langle I; A , g; Y,\pi \rangle$. Given discount factor vector $\bs \delta$, we denote the infinitely repeated game by $G^{\infty} (\bs \delta)$.

For each player $i$, a private history of length $t$, $h_i^t$, is a sequence $(a_i^{(1)},y^{(1)},\ldots, a_i^{(t)},y^{(t)})$ in $H_i^t:= (A_i\times Y)^t$, and the set of $i$'s private histories is $H_i$. A public history $h^t$ is a sequence $(y^{(1)},\ldots,  y^{(t)}) \in H^t:=Y^t$, with the set of all public histories denoted $H$. A behavior strategy $\sigma_i \in \Sigma_i$ maps private histories to $i$'s mixed actions: $\sigma_i:H_i\rightarrow \Delta A_i$. It is public if, it is measurable with respect to $H$. We adopt perfect public equilibrium (PPE) as our solution concept: a strategy profile $\sigma$ is a PPE  if, for all $i$, $\sigma_i$ is public, and for all public histories $h^t$, $\sigma|_{h^t}$ is a Nash equilibrium of the infinitely repeated game (under some payoff criterion).

Definition \ref{defi} is extended the obvious way. A strategy profile $\sigma \in \Sigma$ is a  \textbf{Blackwell PPE} (above $\underline{\delta}$) if there exists $\underline{\delta}\in[0,1)$ such that $\sigma$ is a PPE of $G^{\infty} (\bs \delta)$ at any $\bs\delta \geq\underline{\delta}\cdot \bs (1,\ldots,1)$. A vector $v\in  \mathbb{R}^n$ is a \textbf{Blackwell PPE payoff} at $\bs\delta$ if there exists a Blackwell PPE $\sigma$ above some $\underline{\delta}$, with
$\bs\delta\geq\underline{\delta}\cdot (1,\ldots,1)$, such that $v=U(\sigma;\bs\delta)$, where as before $U(\sigma;\bs\delta)$ is the equilibrium payoff vector under $\sigma$ given $\bs\delta$.

The next proposition shows that under imperfect public monitoring, the Blackwell criterion more severely restricts the set of action profiles that can be played in equilibrium. It assumes there is no PRD. Its proof, as well as the proofs of other results in this section,
appear in Appendix B.


\begin{proposition}\label{prop:generic}Fix $I$, $A$, and $Y$. For almost all $(g,\pi)$, given any Blackwell PPE $\sigma$, $\sigma(h^t)$ is either a pure action profile or a stage-game Nash equilibrium, for all $t$, $h^t \in H^t$.
\end{proposition}

That is, the set of reward functions and monitoring structures such that, in some Blackwell equilibrium, after some public history, players choose an action profile that is not pure or a stage-game Nash equilibrium, has measure zero. Other action profiles, even those satisfying myopic indifference, cannot arise. The selected action can no longer be simply ignored.  To understand why, suppose that, after some history, some player, say $i$, is playing a (nondegenerate) mixed action $\alpha_i$ satisfying myopic indifference, yet some player $j$ (perhaps $i$ himself) is not playing a best-reply to $\alpha_{-j}$. By definition, actions $a_i,a_i' \in \supp \alpha_i$ yield the same reward. However, they induce different distributions over public signals, in general. Since $j$ isn't playing a best-reply in the stage game, the continuation strategy profile must depend on the public signal. This typically affects player $i$'s continuation payoff, and hence, his preference between $a_i$ and $a'_i$ in the repeated game.

\begin{proof}
The proof is divided into three parts; we define a non-generic set of reward functions, then a non-generic set of signal distributions, and show that, for any $(g,\pi)$ outside of this set, a Blackwell PPE specifies a pure action or Nash profile at every history.

Generically, a finite game possesses finitely many Nash equilibria (Harsanyi, 1973a). Because $\mathcal{A}^{\textrm{MI}}$ is  the union of sets of Nash equilibria over finitely many games (defined by the possible subsets of actions), there exists a subset $\mathcal{G} \subset \mathbb{R}^{I\times A}$ of measure zero such that, for any $g\in \mathbb{R}^{I\times A} \setminus \mathcal{G}$, the set $\mathcal{A}^{\textrm{MI}}_g$ (the subscript referring to the  reward)  is finite and has distinct payoffs, i.e. $\alpha, \alpha^\prime \in \mathcal{A}^{\textrm{MI}}_g$ implies that for all $i$, $g_i(\alpha)\neq g_i(\alpha^\prime)$.

Fix $g \in \mathbb{R}^{I\times A} \setminus \mathcal{G}$, a period $T>1$, and a $T$-period strategy profile $\sigma^{T} \colon U_{t=0}^{T-1}H^{t} \to \mathcal{A}^{\textrm{MI}}_g$  in which period-$T$ play varies with the first-round signal, \textit{i.e.},
\begin{equation}
  \exists y_1,y_1',y_2,\ldots,y_{T-1} \in Y \textrm{ s.t. }
  \sigma^{T}(y_1,y_2,\ldots,y_{T-1}) 
  \neq \sigma^{T}(y_1',y_2,\ldots,y_{T-1}).
  \label{eqn:non-trivial-dep}
  \end{equation}
For each $i\in I$, $a_i',a_i''\in A_i$ with $a_i'\neq a_i''$ let
  \begin{equation}
\Pi_{g}^{i,T,a_i',a_i'',\sigma^{T}} :=
\left\{ \pi \in (\Delta Y)^A \ \bigg|\ \mathbb{E}[v_i^T  \mid  a_i',\sigma^T]
= \mathbb{E}[v_i^T \mid a_i'',\sigma^T] \right\},
  \label{eqn:non-trivial-dep2}
  \end{equation} 
where $v_i^T = g_i\circ\sigma^{T}\mid_{Y^{T-1}}$ is player $i$'s reward in round $T$. Both sides of the equality in \eqref{eqn:non-trivial-dep2} are polynomials of $\left\{\pi(y\mid a) \mid y \in Y\setminus \{y_0\},a\in A\right\}$, where $y_0$ is an arbitrary signal (we omit  $y_0 \in Y$ because  the probability distribution $\{\pi(y\mid a) \mid y \in Y \}$  adds up to one). Since $g$ features distinct rewards, player $i$'s last-round reward varies with the first-round signal, and hence the polynomials are distinct. As the set of zeros of a non-zero polynomial, the set  $\Pi_{g}^{i,T,a_i',a_i'',\sigma^{T}}$ is of measure zero (Caron and Traynor, 2005; Neeb, 2011). Define
$$
\Pi_g := \bigcup_{i,T,a_i',a_i'',\sigma^T}\Pi_{g}^{i,T,a_i',a_i'',\sigma^{T}},
$$
where $a'_i,a_i'' \in A_i$, $a_i' \neq a_i''$ and $\sigma^T$ varies with the first-round signal. $T$ runs over a countably infinite set, for each element of which $\sigma^T$ runs over a finite set because $ g \notin \mathcal{G}$ ensures that the range $\mathcal{A}^{\textrm{MI}}_g$ of $\sigma^T$ is finite. Therefore, $\Pi_g$ has measure zero.

Consider a game $G$ with $g \notin \mathcal{G}$. Let $\sigma$ be a Blackwell equilibrium that prescribes an action profile that is neither pure nor a stage-game Nash equilibrium after some public history; without loss of generality, after $h^0$.  By hypothesis, there is a player $i\in I$ who mixes over (at least) two  distinct actions $a'_i$, $a''_i$, and a player $j$ who does not use a stage-best-reply to $\sigma_{-i}(h^0)$. Then there exists an earliest round $T^* \in \mathbb{N}$ in which the action profile depends  on the first-round signal; otherwise, $j$ has a profitable deviation in the initial round.
Since $i$ mixes at the initial history, the  payoff conditional on playing $a_i'$ is equal to the payoff following $a_i''$, for all $\delta$ in an open interval $\mathcal{O}\subset (0,1)$; as in Proposition \ref{prop:X}, it implies that the payoff at each round $t$ is the same. In particular, applying this to $T^*$, and denoting by $\sigma^{T^*}$ the $T^*$-period truncation of $\sigma$, this
implies that $\pi \in \Pi_{g}^{i,T^*,a_i',a_i'',\sigma^{T^*}} \subset \Pi_g$.
%
%
\end{proof}


Again, an immediate implication of  Proposition \ref{prop:generic} is a lower bound on equilibrium payoffs. Recall that $\underline {v}_i^{\textrm{pure}}$ is $i$'s pure minmax, and $\underline{v}_i^{\textrm{NE}}$ is $i$'s worst Nash payoff.
\begin{corollary}\label{cor:mmax}
Fix $A,I$ and $Y$. For almost all $(g,\pi)$, every Blackwell equilibrium payoff $v$ satisfies
$v_i \geq \min\{\underline{v}_i^{\mathrm{pure}},\underline{v}_i^{\textrm{NE}}\},$ for all $i\in I$.
\end{corollary}

\subsection{A ``Folk'' Theorem}
Given Proposition \ref{prop:generic} and Corollary \ref{cor:mmax}, a folk theorem under the Blackwell criterion must involve a smaller payoff set, and stronger assumptions than those imposed by FLM. First, player $i$'s equilibrium payoff is bounded below by his pure minmax, or his lowest Nash equilibrium payoff, whichever is lower.

Second, in general, mixed actions can help detect deviations, or discriminate among them. Hence, the identifiability assumptions  must be strengthened.
\begin{defn}
The monitoring structure $(Y,\pi)$ satisfies \textbf{pairwise full rank} for a profile $\alpha$ if for all $i,j\in I$ with $i\neq j$, the matrix
$\Pi_{i,j}(a)$ whose rows are $\{\pi^\top(\cdot\mid a_i',\alpha_{-i})\mid a_i'\in A_i\} \cup \{\pi^\top(\cdot\mid a_j',\alpha_{-j})\mid a_j'\in A_j\}$ has rank
$|A_i| + |A_j| - 1$.
\end{defn}

Denote
\[
F^*:=\{v \in F \mid v_i \ge \min\{\underline{v}_i^{\mathrm{pure}} , \underline{v}_i^{\textrm{NE}}\}\  \forall i\in I\}.
\]





\begin{thm}
\label{thm:NoPRD} Fix $I,A$ and $Y$. Suppose that $\pi$ satisfies pairwise full rank for all pure action profiles. Suppose also that there exists $a^* \in A$, $Y^*\subset Y$ such that  $\pi(Y^{*} \mid a_{-i}^{*},a_{i})<\pi(Y^{*}\mid a^{*})<1$, $\forall i\in I,~a_{i}\not =a^\ast_i$.
Then for any $ v\in \interior  F^*$, there exists  $\underline \delta <1$ such that for all $\delta \in (\underline \delta,1)$, $v$ is a Blackwell PPE payoff at $\delta$.
\end{thm}

 The assumption that a pair $(a^*,Y^*)$ as stated in the theorem exists is technical; while it is relatively mild, it is needed in the proof, and we do not know whether some version of it is necessary for the result. It allows the players to emulate a PRD (and could be dispensed if a PRD was assumed).%
 \footnote{That is, the event $\{y \mid y \in Y^{*}\}$ is a public, binary random variable, whose likelihood is maximized if players use $a^*$. Hence, it suffices to make its occurrence desirable to ensure that players are willing to generate that signal.}
\medskip

Proving the theorem involves several steps. It is more instructive
to explain some of them in the special case in which monitoring takes a product structure, and assuming a
PRD. Hence, the proof of Theorem \ref{thm:NoPRD}
should be read after
the proof of Theorem \ref{thm:product} and the discussion at the end of this section.

\subsection{Public Randomization\label{subsec:Prod}}
If we now assume a PRD (a uniform draw from the unit interval), in the special case in which monitoring has a product structure, some non-Nash myopically indifferent mixed action profiles may be used. Hence, the Blackwell payoff set may exceed that of Theorem \ref{thm:NoPRD}. We say that $(Y,\pi)$ has a \textbf{product structure} if
$$Y = \prod_i Y_i\mbox{ and }
\pi(y\mid a) = \prod_i \pi_i(y_i\mid a_i),$$
 where $\pi_i(\cdot \mid a)$ is the marginal distribution of $\pi(\cdot\mid a)$ on $Y_i$.

The relevant lower bound on player $i$'s equilibrium payoff is given by the solution to the following program.


\begin{equation*}
(P_i^{\mathrm{MI},\pi_i}): \hspace{1cm}
\min_{\alpha \in \mathcal A^{\mathrm{MI}},  x_i\colon Y_i \to \mathbb{R}}
\left\{ g_i(\alpha) + \sum_{y_i\in Y_i}\pi_i(y_i\mid \alpha_i)x_i(y_i) \right\}
\label{v_prd_ps}
\end{equation*}
subject to
\begin{align}
g_i(\alpha) + \sum_{y_i\in Y_i}\pi_i(y_i\mid \alpha_i)x_i(y_i)
&\geq g_i(a_i,\alpha_{-i}) + \sum_{y_i\in Y_i}\pi_i(y_i\mid a_i)x_i(y_i) &&\forall a_i\in A_i,\\
x_i(y_i)&\geq 0 &&\forall y_i\in Y_i. \label{x_geq_0}
\end{align}
Let $\underline v_{i}^{\mathrm{MI},\pi_i}$ denote the minimum. We note that, since it is feasible to pick  a stage-game Nash equilibrium for $\alpha$, and to set $x_i(\cdot)=0$, $\underline v_{i}^{\mathrm{MI},\pi_i} \le \underline v_{i}^{\mathrm{NE}}$. Also, since $\mathcal A^{\mathrm{MI}}$ can be strictly larger than $A$, it is easy to find games such that $\underline v_{i}^{\mathrm{MI},\pi_i} <  \min\{\underline{v}_i^{\mathrm{pure}} , \underline{v}_i^{\textrm{NE}}\}$.

This program is nothing but the ``scoring algorithm'' (in the direction that minimizes $i$'s payoff) introduced by Fudenberg and Levine (1994), with the restriction that players $-i$ are constrained to choose from $\mathcal A^{\mathrm{MI}}$. Indeed, this constraint must be satisfied in a Blackwell equilibrium (given that this is already the case under perfect monitoring, see Proposition \ref{prop:X}, this should come as no surprise). It immediately follows that $\underline v_{i}^{\mathrm{MI},\pi_i}$ is a lower bound on $i$'s equilibrium payoff.

For this bound to be tight, a rank assumption is needed, for which we follow FLM.
\begin{defn}
A profile $\alpha$ satisfies \textbf{individual full rank} (IFR) if for all $i$ the vectors $\{\pi(a_i,\alpha_{-i}) \mid a_i\in A_i\}$  are linearly independent.
\end{defn}
Let
$$F^{\mathrm{MI}, \pi}
:= \{v\in \co (u(A)) \mid v_i \geq \underline{v}_i^{\mathrm{MI},\pi_i} \, \forall i\in I\}.$$
 The characterization is the following.
\begin{thm}\label{thm:product}
Assume a PRD. Suppose $(Y,\pi)$ has a product structure. Every Blackwell PPE payoff vector $v$ satisfies $v_i \geq \underline v_i^{\mathrm{MI},\pi_i}$ for all $i\in I$.\\
Conversely, if $\pi_i$ satisfies IFR for all $i$, then for any  $v\in \interior F^{\mathrm{MI}, \pi}$, there exists $\underline \delta<1$ such that for all
$\delta \in (\underline \delta,1)$, $v$ is a Blackwell PPE payoff vector at $\delta$.
\end{thm}


The proof of the necessity part matches the proof of Proposition \ref{prop:X}. The sufficiency part of the theorem has a two-step proof:
\begin{itemize}
 \item[1.] We show that at \textit{some} $\delta_0<1$ there is a \textit{robust equilibrium}, \textit{i.e.}, a strategy profile that is a PPE in a neighborhood of $\delta_0$ with payoff vector $v$ at $\delta_0$.
 \item[2.] We use the PRD to periodically restart (or ``reboot'') the game, discarding the history up to that point, which allows us  to lower the discount factor at which incentive compatibility must be checked.\footnote{This is somewhat in the spirit of what is known as ``Ellison's trick'' (Ellison, 1994).}
 That is, if $\sigma$ is a PPE at some $\delta_0 \in (0,1)$, we can construct a related  equilibrium at higher discount factors by rebooting appropriately to reduce the effective discount factor to $\delta_0$.
\end{itemize}


\subsubsection{Rebooting}

We begin with the second step, {\textit {rebooting}}. Fix $\sigma$ and $ p \in (0,1)$. Let $\sigma^p$ denote a strategy profile that follows $\sigma$ but reboots the game with probability $p$ at the end of each round, independently across rounds; \textit{i.e.}, if at the end of some round the value of the PRD  is less than $p$, we discard the history and restart playing $\sigma$. More precisely, let $\omega_1,\omega_2,\ldots$ be the sequence of PRD draws; then, given $(y_1,\omega_1,\ldots,y_t,\omega_t)$, we let
$$
\sigma^p(y_1,\omega_1,y_2,\omega_2,\ldots,y_t,\omega_t) := \sigma(y_{\tau+1},\ldots,y_t),
\mbox{ where } \tau := \max\{s \leq t \mid \omega_s \leq p\},$$
with the convention that $\max \emptyset = 0$.

If a player has a discount factor $\delta$, the payoff stream from $\sigma^p$ is evaluated at discount $\delta$; however, a simple calculation shows that a player's incentive to deviate from $\sigma^p$ at discount $\delta$ is the same as his incentive to deviate from $\sigma$ at discount $\delta (1-p)$.


This definition and the preceding discussion imply the following lemma, which essentially reduces a global robustness problem to a  local one.
\begin{lemma}[Reboot Lemma] \label{lemma:reboot}
If $\sigma$ is a PPE for all discount factors in some interval $(\delta_0 - \Delta, \delta_0 + \Delta) \subset [0,1)$, for some $\Delta>0$, then $\sigma^p$ is a Blackwell PPE above $\frac{\delta_0 - \Delta}{1-p}$ for $p \in (0,1)$ such that $\frac{\delta_0 - \Delta}{1-p} < 1 < \frac{\delta_0 + \Delta}{1-p}$, and $$ U(\sigma^p, \delta_0/(1-p)) = U(\sigma, \delta_0).$$
\end{lemma}

\subsubsection{Constructing a Robust Equilibrium}

Given the Reboot Lemma, existence of the desired Blackwell equilibrium follows from the construction of a robust equilibrium. Existence of a robust equilibrium will be demonstrated by adapting arguments from Abreu, Pearce, and Stacchetti (1990) (henceforth APS 1990). Robustness requires incentives to hold for a range of discount factors; this motivates a stronger notion of self-generation than proposed in APS 1990. We provide some intuition for the two ways in which our definition needs to be stronger. The crux is that varying the discount factor, however slightly, could affect incentives.

If the current action profile is pure, changing the discount factor may reverse weak incentives to not deviate. This can be tackled by giving strict incentives (a slack of at least $(1-\delta)\eta$ in the definition below) not to deviate.

The problem is subtler when mixed actions are needed, such as when delivering $\underline v_{i}^{\mathrm{MI},\pi_i}$ to player $i$.
Suppose  $\alpha, x$ solve Program $P_i^{\mathrm{MI},\pi_i}$.
Since  $\alpha \in \mathcal A^{\mathrm{MI}}$, the current payoff $g_i(\alpha)$ is constant on $\supp(\alpha_i)$.
Distinct actions in the support of $\alpha_i$ could generate different distributions over various continuation payoff vectors $w$, but with the same expected value for the $i$-th component $w_i$.
But even if these actions induce the same distribution over $w_i$ at $\delta$, if they induce different distributions over continuation payoff vectors, they could induce different distributions over action paths. These could give different expected continuation payoffs for $i$ at discount factors even slightly away from $\delta$. To circumvent this, we need to ensure that any two such actions induce the same distribution over continuation payoff vectors, and hence the same probability distribution over action paths.

To this end we use the randomization device - a uniform draw (denoted $\nu$) from $[0,1]^n$ - to ``garble'' the distribution of continuation payoff vectors.

\begin{defn} For any $\eta >0$, $W \subset \mathbb{R}^n$, and $\delta \in (0,1)$, the set  $\mathcal B_{\eta}(W;\delta)$  comprises points $v \in \mathbb{R}^n $ such that  $v=(1-\delta) g(\alpha) + \delta \, \mathbb{E}(w\mid \alpha)$ for some continuation payoff function $ w: Y \times [0,1]^n \rightarrow W $
taking finitely many values, and a current action profile $\alpha$ that is a NE of the normal-form game
with payoffs $(1-\delta) g(a) + \delta \, \mathbb{E}(w\mid a)$, under the additional condition that 
for any $i \in I$, for any $a^\prime_i\in A_i$, at least one of the following is true:
\begin{align}
w|\alpha & \overset{d}{=}   w|a^\prime_i,\alpha_{-i} \label{eq:wdist} \\ 
v_i & \geq (1-\delta) g_i(a^\prime_i,\alpha_{-i}) + \delta \mathbb{E}_{\nu} \sum_{y\in Y}\pi(y\mid a^\prime_i,\alpha_{-i})w_i(y, \nu)
+ (1-\delta)\eta.
\label{eq:strictIC}
\end{align}
For any $\eta >0$, the mapping
$\mathcal B_{\eta}: 2^{\mathbb{R}^n} \times [0,1) \rightarrow 2^{\mathbb{R}^n}$ is called an `\textbf{$\eta$-strong APS mapping}'.
\end{defn}

The slack in (\ref{eq:strictIC}) does depend on the discount factor. Alternatively, we can view it as a slack of $\eta$ in un-normalized or total payoffs. Our definition requests that all pure actions on the support of $\alpha_i$ output the same distribution of continuation payoffs as in \eqref{eq:wdist}. Moreover, a deviating pure action $a_i$ off the support of $\alpha_i$ either (a) satisfies \eqref{eq:wdist} so that it leads to the same distribution over continuations, and is therefore unprofitable at any discount given that $\alpha$ is a Nash of the auxiliary game; or (b) satisfies \eqref{eq:strictIC} and therefore entails a loss of at least $(1-\delta)\eta$, which will allow us to show it is unprofitable at discount factors in some neighborhood of $\delta$.  Our  strengthened notion of self-generation follows.

\begin{defn}
For $\eta >0$, a set $ W \subset \mathbb{R}^n$ is said to be \textbf{$\eta$-strong self-generating} 
at $\delta$  if  $ W \subset \mathcal B_{\eta}(W;\delta) $ for a strong APS mapping $\mathcal B_{\eta}$.
\end{defn}

At this point the standard approach shows that any smooth set in the interior of the feasible and individually rational set is self-generating.
Our proof differs in three ways from this. First, we use the notion of $\eta$-strong self-generation, to leave ``wiggle room'' for varying discounting and achieving robustness. Second, we show this property for closed balls, rather than directly for all smooth sets; this is analytically more tractable. Points in the interior of a ball are generated by playing a Nash equilibrium of the stage game; the required continuations are in the ball for high enough discount factors.  Boundary points are harder --if the action profile chosen provides weak incentives, we ``mix in'' a small probability of continuation payoffs that provide strict incentives; these exist by individual full rank.
Third, to satisfy condition (\ref{eq:wdist}) we modify each player's continuation payoffs to take only the two extreme values, the only variation being the probabilities with which these two values are chosen for various $y \in Y$; these stochastic continuation values require a public randomization device to carry out ``garbling.''

Following FLM 1994, we prove a local self-generation property, and then leverage this property to prove we can construct robust equilibria.

\begin{lemma} \label{lem:ballSS}
If $\pi$ satisfies IFR, for any $c \in F^{\mathrm{MI}, \pi}$ and $r > 0$ satisfying $B(c,r) \subset \interior(F^{\mathrm{MI}, \pi})$ there exist $\underline \delta<1$ and $\eta>0$ for which the closed ball $B(c,r)$  is $\eta$-strongly self-generating at any $\delta \geq \underline{\delta}$.
\end{lemma}

The subsequent Lemma \ref{lem:RobustEq}, which is essentially a 
locally robust PPE folk theorem, follows immediately from Lemma \ref{lem:ballSS}, the continuity of payoffs in $\delta$, and familiar self-generation arguments \`a la APS 1990 and FLM 1994.

\begin{lemma}[Robust Equilibrium Exists] \label{lem:RobustEq}
With a product monitoring structure and a PRD, if $\pi$ satisfies IFR, for any compact $W \subset \interior (F^{\mathrm{MI},\pi})$ there exists a $\underline\delta\in (0,1)$, a $\eta>0$ and a $X\supset W$ such that for all $\delta\in (\underline \delta,1)$ we have $X\subset \mathcal B_{\eta}(X;\delta)$. Moreover,for any $\delta\in(\underline \delta,1)$ there is a strategy profile $\sigma$ such that $v = U(\sigma; \delta)$ and $\sigma$ is a PPE at all $\delta' \in (\delta-\Delta, \delta+\Delta)$ for some $\Delta>0$.
\end{lemma}

Theorem \ref{thm:product} follows from Lemmata \ref{lemma:reboot} and \ref{lem:RobustEq}. Proofs of Lemmata \ref{lem:ballSS} and \ref{lem:RobustEq} are in Appendix B.

\subsubsection{From here to Theorem \ref{thm:NoPRD}}
\label{ss:NoPRD Proof}

The full proof is contained in Appendix B, to which this section serves as a roadmap. Once again, we use a suitable version of self-generation to create robust equilibria at discount factors away from unity, and then use rebooting to create a Blackwell equilibrium out of these.
Generating a robust equilibrium involves  the same steps
as before. Complications arise from having to emulate a PRD through the public signal. Our assumption
on the existence of a suitable $(a^*,Y^*)$ guarantees that, with suitable rewards, all players can be induced to
maximize the probability of $Y^*$. ``Test phases'' during which players play $a^*$ are therefore introduced to
generate randomness for resets. Several difficulties arise;
we tackle these in turn.

First, unlike the exogenous PRD of Theorem \ref{thm:product},
our tests are open to manipulation. In that theorem, some players may
desire a reset, others not; but they are not given a choice.
Here, we need to provide incentives to players to choose the test action.
To offer rewards and yet still wipe the history, we use two self-generating
payoff sets, a ``punishment'' and a ``reward'' one; every payoff vector in the former
is Pareto-dominated by every payoff vector in the latter. By maximizing the probability of $Y^*$
during a test phase, players maximize the probability of getting to --or staying in-- the `reward' payoff set.

Second, to incentivize players as above, ``normal'' phase rounds --during which punishments and rewards are dished out-- must be far more common than ``test'' phase rounds. The Reboot Lemma only applies when there is exactly one test between two normal-phase rounds linked by incentives --but we necessarily have more normal phase rounds between punishment phases.
To fix this, we break up play into ``cycles.'' Thus if we have $N$ rounds in the normal phase and $T+1$ rounds of testing, rounds within a cycle are separated by $T+N+1$ rounds. Rounds that form part of the same normal phase are each linked to a different cycle, with continuation along that cycle coming $T+N+1$ rounds later. In other words, we play $N$ copies of the strategies given by self-generation, punctuated by test phases of length $T+1$ each. The outcome (pass vs.~fail) of each test phase affects every cycle identically.

Third, since the test is constructed to meet the need for a certain passing
probability, it may give payoffs different than those we wish to deliver in equilibrium.
As $T$ is given by this need, we keep it fixed and raise both
$N$ and the permissible discount factor so that overall payoffs are largely due to normal phase play. Then, payoffs from the normal
phase are adjusted to take into account expected payoffs during the test
phase so as to deliver an overall expected payoff equal to the target payoff.

Fourth, at the same time, the original robust equilibrium's incentives should apply; so $\delta^{T+N+1}(1-p)$, where $p$ is the reset probability, must lie in the robust range.

\section{Imperfect Private Monitoring}\label{sec:PM}

This section turns to  imperfect private monitoring. Again, we are given a finite set of players $I=\{1,2,...,n\}$, 
for each  $i \in I$ a finite set of actions $A_i$ and a reward function $g_i:A \rightarrow \mathbb{R}$. 
A (private) monitoring structure is a pair $(Y, \pi)$, with $Y=\prod_{i \in I}Y_i$, finite, and  $\pi\colon A\to \Delta Y$ mapping  $a \in A$ into the probability that the signal profile $(y_1,\ldots,y_n) \in Y$ obtains. Player $i$ only observes $y_i$. Let $ G= \langle I; A , g; Y,\pi \rangle$. Given discount factor vector $\bs \delta$, we denote the infinitely repeated game by $G^{\infty} (\bs \delta)$.

A $t$-length private history $h_i^t$ is a sequence $(a_i^{(1)},y_i^{(1)},\ldots, a_i^{(t)},y_i^{(t)}) \in H_i^t$.  The set of all private histories for $i$ is denoted $H_i$.
A behavior strategy $\sigma_i \in \Sigma_i$ maps private histories to mixed actions, $\sigma:H_i \rightarrow \Delta A_i$. We follow the literature by adopting sequential equilibrium as  solution concept.

Definition \ref{defi} is extended the obvious way. A strategy profile $\sigma \in \Sigma$ is a  Blackwell equilibrium (above $\underline{\delta}$) if there exists $\underline{\delta}\in[0,1)$ such that $\sigma$ is a (sequential) equilibrium of $G^{\infty} (\bs \delta)$ at any $\bs\delta \geq\underline{\delta}\cdot \bs (1,\ldots,1)$. A vector $v\in  \mathbb{R}^n$ is a Blackwell equilibrium payoff at $\bs\delta$ if there exists a Blackwell equilibrium $\sigma$ above some $\underline{\delta}$, with
$\bs\delta\geq\underline{\delta}\cdot(1,\ldots,1)$, such that $v=U(\sigma;\bs\delta)$, where as before $U(\sigma;\bs\delta)$ is the equilibrium payoff vector under $\sigma$ given $\bs\delta$.


Our focus is on games in which $\mathcal{A}^\textrm{MI}=A$. An important property of such games is the following.
\begin{lemma}\label{unique_a}
If $\mathcal{A}^\textrm{MI}=A$, the stage game has a unique and pure Nash equilibrium.
\end{lemma}

\begin{proof}
Since $\mathcal{A}^\textrm{MI}=A$, no player has a tie against any pure action profile of his opponents. 
Suppose, to the contrary, that there exist a player $i$, a profile of actions of the others $a_{-i}$, and two distinct actions $a_i$ and $a_i'$ of $i$ satisfying  $g_i(a_i,a_{-i}) = g_i(a_i',a_{-i})$.
Then, $(\frac12 a_i + \frac12 a_i',a_{-i})$ would belong to $\mathcal{A}^\textrm{MI}$. Thus, every stage-game Nash equilibrium would be strict, and hence pure. By the index theorem, such an equilibrium is unique.
\end{proof}


 The prisoner's dilemma satisfies $\mathcal{A}^\textrm{MI}=A$. So does the product choice game, and 2x2 dominance solvable games. More generally, $\mathcal{A}^\textrm{MI}=A$  if there is no tie against any pure action profile of the opponents, and for each player $i$, the ordinal ranking over $i$'s actions is independent of $(a_{i+1},\ldots,a_n)$. 
The converse of Lemma \ref{unique_a} does not hold: there exist games with a unique and pure Nash equilibrium, yet $\mathcal{A}^\textrm{MI}\neq A$. (For instance, dominance solvability is not enough in general.)

We focus on a special class of monitoring structures. Let $\pi_i(y_i \mid a):=\sum_{y_{-i} }\pi(y\mid a)$.
\begin{defn}
A monitoring structure $(Y, \pi)$ is \textbf{conditionally independent} if  
$$\pi(y\mid a) = \prod_i \pi_i(y_i\mid a) \mbox{ for all } a\in A,  y \in Y.$$
\end{defn}
Conditional independence is a special and admittedly non-generic property. Yet, it plays an important role in the literature. In particular, for the prisoner's dilemma, Matsushima (2004) establishes a folk theorem under conditional independence.

The monitoring structure $(Y, \pi)$ has \textbf{full support} if $\pi_i(y_i\mid a)>0$ for all $i\in I$, $y_i \in Y_i$, $a \in A$.



\subsection{An ``Anti-Folk'' Theorem}

\begin{thm}\label{thm:anti}
Suppose that $\mathcal{A}^\textrm{MI}=A$, and that $(Y, \pi)$ is conditionally independent and has full support. The unique Blackwell equilibrium outcome is the repetition of the stage-game Nash equilibrium.
\end{thm}

\begin{proof}
Let $\sigma$ be a Blackwell equilibrium.
First, we show that $\sigma$ is pure and history-independent on the equilibrium path. Let $t$ be the first round at which $\sigma$ prescribes either mixed or history-dependent actions, if it exists. Since players play pure and history-independent actions until round $t-1$, the conditional independence of the monitoring structure implies the independence of players' private histories at the beginning of round $t$, $(H_t,p_t)$ with $p_t = p_{1t}\times \cdots\times p_{nt}$. Note that each player $i$ is indifferent among all continuation strategies against the opponents' continuation strategies. By an argument similar to that in Proposition \ref{prop:X}
(the Identity/Uniqueness Theorem), player~$i$ is indifferent among all actions played with positive probability at round $t$ against the opponents' actions in the same round. Let $\overline{\alpha}_i(a_i) = \sum_{h^t_{i}} p_{it}(h_{it})\sigma_i(h^t_{i})(a_i)$ for each $i$ and $a_i$. Then $(\overline{\alpha}_1,\ldots,\overline{\alpha}_n) \in \mathcal{A}^\textrm{MI}$. This contradicts $\mathcal{A}^\textrm{MI}=A$.

Since the monitoring structure has full support, there is no reason to play a suboptimal action against the opponents' history-independent strategies. Thus, players play the unique stage-game Nash equilibrium  on the equilibrium path.
\end{proof}

The intuition for this result relies on Matsushima (1991). Indeed, the first step of the proof of Theorem \ref{thm:anti} follows his. He shows inductively that a pure strategy  satisfying ``independence of irrelevant information'' must be history-independent. Hence, non-trivial equilibria must involve some indifference across some actions, for some player, after some private history. This is inconsistent with $\mathcal{A}^\textrm{MI}=A$, given that $\sigma$ must be a Blackwell equilibrium.




\medskip

If $\mathcal{A}^\textrm{MI} \neq A$, \textit{i.e.}, there exists $\alpha \in \mathcal{A}^\textrm{MI} \setminus A$, then players may play $\alpha$ at, say, round $1$, so that independence of private histories fails. Even if players have played pure actions so far, they can play possibly history-dependent actions at round $t$ so long as their ``averages'' are equal to $\alpha_i$. By using history dependence appropriately, one can engineer non-myopic equilibrium behavior at earlier rounds, 
see the following example.

\begin{figure}
\[
\begin{array}{ccc}
& H & D \\ \cline{2-3}
\multicolumn{1}{c}{H} & \multicolumn{1}{|c}{0,0} & \multicolumn{1}{|c|}{5,1} \\ \cline{2-3}
\multicolumn{1}{c}{D} & \multicolumn{1}{|c}{1,5} & \multicolumn{1}{|c|}{4,4}
\\ \cline{2-3}
\end{array}
\]
\caption{The game in Example 2.}\label{example22}
\end{figure}

\begin{example}
\label{Ex2}
Consider the repetition of the hawk-dove game given by the payoff matrix in Figure 2.
The conditionally independent monitoring structure $(Y,\pi)$ is given by  $Y_1 = Y_2 = \{h,d\}$, and
$$
\pi_i(y_i = h \mid a_j = H) = \pi_i(y_i = d \mid a_j = D) = 0.9.
$$
The stage game has three Nash equilibria: $(H,D)$, $(D,H)$, and $(\frac12 H + \frac12 D,\frac12 H + \frac12 D)$,
hence $\mathcal{A}^\textrm{MI} \neq A$.
Consider the following symmetric strategy profile: at every odd round, play $D$, and at every even round, play $H$ if one's own previous action is $D$ and the signal is $h$, play $\frac{5}{9}D + \frac{4}{9}H$ if one's own previous action is $D$ and the signal is $d$, and play $D$ if one's own previous action is $H$ (off the equilibrium path). Note that at an even round on the equilibrium path, each player plays
$$0.1 \times H + 0.9 \times \left(\frac{5}{9}D + \frac{4}{9}H\right) = \frac12 H + \frac12 D$$
on average, and no player has an incentive to deviate. If a player deviates to $H$ at an odd round, he receives a reward of $5$ in that  round, and at the next round, faces $0.9 \times H + 0.1 \times \left(\frac{5}{9}D + \frac{4}{9}H\right) = \frac{17}{18}H + \frac{1}{18}D$ on average, and so receives a payoff of $\frac{21}{18}$. Since
$$
4 + \delta \times \frac{5}{2} \geq 5 + \delta \times \frac{21}{18} \Leftrightarrow \delta \geq \frac{18}{23},
$$
this strategy profile is a Blackwell equilibrium.
\end{example}




\section{Extension: Limit Blackwell Payoffs}\label{sec:Ext}


\color{black}
Rather than fixing a discount factor, and characterizing the set of Blackwell equilibrium payoffs evaluated at this discount factor,
one might wonder what payoffs can be achieved as limit payoffs of some Blackwell equilibrium.




\begin{defn}
A payoff vector $v$ is a \textbf{limit Blackwell payoff}  if there exists a Blackwell SPNE  $\sigma$ such that $U(\sigma; \bs\delta) \rightarrow v$ as $\bs\delta \rightarrow \bs (1,\ldots,1)$.
\end{defn}


As one might surmise, the set of such payoffs matches the one that appears in Theorem \ref{thm:perfect}, despite the fact that not all Blackwell equilibria have payoffs that converge as $\bs\delta \rightarrow \bs (1,\ldots,1)$.

\begin{thm}\label{LBPPM}
Fix a repeated game under perfect monitoring, such that the dimension of $F$ is $n$. A payoff vector is a limit Blackwell payoff if it is in $F^{\mathrm{MI}}$, and only if it is in the closure of $F^{\mathrm{MI}}$.
\end{thm}

The proof proceeds as that of Theorem \ref{thm:perfect}. The main difference is that on-path play yields payoffs that converge to the target payoff. To this end, we construct sequences of pure actions that approximate the target payoff, yet preserve individual rationality for low discounting. 

\section{Conclusion}

We apply the Blackwell optimality criterion to repeated games. This restricts equilibrium behavior by ruling out mixed (non-pure) strategies in general, except for particular profiles that depend on the monitoring structure. This restriction on behavior implies bounds on equilibrium payoffs, which reflect and clarify the role that mixed strategies play under different monitoring structures. Under perfect monitoring, they are used during minmaxing. Under imperfect public monitoring, they also help detection. Under private, conditionally independent monitoring, they must be part of any equilibrium that is not the repetition of the stage-game Nash equilibrium.

We focus on discount factor perturbations. Aside from the fact that this is a very natural parameter in studying repeated games, we find it plausible that in some applications agents know the consequences of their actions, but not the discount factors, as the latter might depend on the rates of interest set by a monetary authority. 

As a result, the minmax levels must be adjusted under perfect and imperfect public monitoring, and the identifiability conditions must be strengthened under imperfect public monitoring. With these modifications, folk theorems apply. Finally, under private, conditionally independent monitoring, only the repetition of the stage-game Nash equilibrium survives in prisoner's dilemma-like games.

This paper is very much a first pass. Mixed strategies also play an important role when considering games with short-run vs.\ long-run players (see Fudenberg, Kreps, and Maskin, 1990).\footnote{Here, under perfect monitoring, players are ``hybrid:'' they behave as if they were short-run as far as behavior strategies are concerned, but long-run for pure strategies.} Their importance under general private monitoring also remains to be seen.

Our results provide a somewhat nuanced justification for the skepticism with which mixed strategies are often viewed by empiricists when modeling long-run relationships, and their focus on pure strategies. At the same time, it may be that mixed strategies can be ``purified''  here as well (Harsanyi, 1973b). What equilibria survive  under the Blackwell optimality criterion in a setting that includes random payoff shocks (see Bhaskar, Mailath and Morris, 2008; P\k{e}ski, 2012) is an  open question.


\section*{Appendix A: Proofs for Section 2 (Perfect Monitoring)}
\label{appendix:2}

\subsection*{Proof of Theorem \ref{thm:perfect}}

We start by stating a useful result, which allows construction of action sequences with desired payoffs and continuation payoffs within given bounds.

\begin{thm}[Dasgupta and Ghosh, 2021]
\label{thm:DG2017}
For all $v\in F$ and $\varepsilon>0$,
there exists $\hat{\delta}>0$ such that for any $\delta\geq\hat{\delta}$
there is a sequence of action profiles $(a^{(t)}:t\geq0)=:\, a(v,\varepsilon,\delta)$ such that
\begin{equation} \label{SAactions_v}
 v
 = (1-\delta)
 \sum_{t\geq0}\delta^{t}g( a^{(t)});\;
\left\| (1-\delta)\sum_{t\geq\tau}\delta^{t-1}g(a^{(t)})- v \right\|
\leq \varepsilon \; \forall \tau \geq 1.
\end{equation}\label{DG}
\end{thm}

In words, given an $\varepsilon$ and a high enough discount $\delta$, the discounted payoff of the whole sequence is $v$, while the continuation payoff from any time $\tau$ onwards is $\varepsilon$-close to $v$.
As our purposes require strategies designed without knowledge of the exact discount factors, we need to know how the continuation payoffs of a fixed sequence of actions change as individual discount factors increase. The next lemma answers this by showing that if \textit{all} $\delta$-discounted continuation payoffs of a sequence are bounded above and below,
the same is true at higher discount factors. This helps us get discount robustness.

\begin{lemma}[Patience Lemma] Given $\delta\in (0,1)$, if a sequence of
real numbers $(x^{(t)})_{t\in\mathbb{Z}_{+}}$ satisfies
\begin{equation} \label{eq:PLemma_2bounds}
\underline{x}\leq (1-\delta)\sum_{t=\tau}^{\infty}\delta^{t-\tau}x^{(t)} \leq \bar{x}, \forall \tau \geq0
\end{equation}
for some $\overline{x}$ and $\underline{x}$ in $\mathbb{R}$ and some $\delta\in(0,1)$, 
then the same inequalities (\ref{eq:PLemma_2bounds}) hold for any $\delta^{\prime}\in(\delta,1)$.
\end{lemma}

\begin{proof} Define $f:[\delta,1)\times\mathbb{Z}_{+}\rightarrow\mathbb{R}$
by
\begin{equation}
f(\delta^\prime,\tau) :=(1-\delta^\prime)\sum_{t=\tau}^{\infty}\delta^{\prime t-\tau}(x^{(t)}-\underline{x})\label{start}
\end{equation}
which is {\it ex hypothesi} non-negative when $\delta^\prime=\delta$ for every $\tau\geq 0$. For any $\delta^{\prime}>\delta$, we have after some standard substitutions and simplifications:

\begin{equation}
f(\delta^{\prime},\tau)\;=\frac{1-\delta^{\prime}}{1-\delta}f(\delta,\tau)
+\frac{1-\delta^{\prime}}{1-\delta}(\delta^{\prime}-\delta)
\sum_{t=\tau+1}^{\infty}\delta^{\prime t-\tau -1}f(\delta,t).\label{aux}
\end{equation}
From $\delta^{\prime}-\delta\geq0$ and (\ref{start}) it follows that the right side of (\ref{aux}) is non-negative.
This leads to the inequality on the right side of  (\ref{eq:PLemma_2bounds}) at $\delta^\prime$; the other follows similarly.
\end{proof}

We now give a constructive proof of the positive part of Theorem \ref{thm:perfect}.
\begin{proof}
Fix $v\in F^{\mathrm{MI}}$.
The overall structure follows folk theorems closely---unilateral deviations are followed by minmax punishments, followed by a post-minmax phase that rewards every player who carried out the minmax phase.  
Following Abreu, Dutta, and Smith (1994) we now construct $n$ points that will serve as post-minmax payoffs at a given discount factor; we shall also show that at higher discount factors the actual post-minmax payoffs will be nearby. 

First note that Full Dimensionality implies Abreu, Dutta, and Smith (1994)'s Non-Equivalent Utility; hence we are able to obtain points 
$ \{ x(i) \in F| i=1,2,\ldots, n\}$ satisfying payoff asymmetry (PA): $\forall i,j$ with $i \neq j,  \;x_{i}(i)  < x_{i}(j)$. Let $w(i) \in F$ be the point in $F$ where $i$ gets the lowest feasible payoff, ignoring considerations of individual rationality. 
Now for each pair $(\beta$, $\eta) \in (0,1)^2$ and each $i$, let
\[
 y(i) := \beta (1-\eta) w(i) + \beta \eta x(i)+ (1-\beta)v. 
\]
For suitably small choices of $\beta$ and $\eta$, by construction these points have the following properties for all $i$: (1) strict myopic indifference rationality (SMIR), \textit{i.e.}, $y_j(i) > \underline{v}_j^{\textrm{MI}}$ for all $j$; (2) PA; (3) target payoff dominance (TPD), \textit{i.e.}, $y_i(i) < v_i$.  
Since all inequalities are strict, take any $\varepsilon>0$ such that all the above inequalities hold with a slack of $3\varepsilon$. 

Each $y(i)$ is generated by a convex combination of the pure-action payoffs $\{g(a)| a\in A\}$, so $y(i)\in F$. We approximate each $y(i)$ within $\varepsilon$ by a rational convex combination of the pure payoff points. Without loss of generality we can use these weights to construct sequences $(\widetilde  a^i_t)_{t=0}^{t=T-1}$ of the same length $T$ such that 
\begin{equation}
\label{e-approx}
 \left\| \frac{1}{T}\sum_{t=0}^{T-1} g(\widetilde  a^i_t) - y(i) \right\|< \varepsilon.  
\end{equation}

Defining
\begin{equation}\label{Rewardssetup}
v(i):=\frac{1}{T}\sum_{t=0}^{T-1} g(\widetilde  a^i_t)\in F^{\mathrm{MI}},
\end{equation}
we have obtained points $ \{ v(i) \in \mathbb{R}^n| i=1,2,\ldots, n\}$  and for each $i$ a finite sequence of action profiles $\widetilde a^{i} = (\widetilde a^{i}_t)_{t=0}^{t=T-1}$ suitably reordered so that the payoff of $i$ is increasing along the sequence for $i$ (\textit{i.e.}, $g_i(\widetilde a^{i}_t)<g_i(\widetilde a^{i}_{t^\prime})$ whenever $t<t^\prime$). 
Since the SMIR, PA, and TPD constraints for the $\{y(i)\}$ all held with a slack of $3\varepsilon$, (\ref{e-approx}) implies that these constraints will hold with a slack of at least $\varepsilon$ for $\{v(i)\}$:
\begin{align}
\forall i, \; \underline v_i^{\mathrm{MI}}+\varepsilon & <v_{i}(i), \label{idev}\\
\forall i, \; v_i(i)+\varepsilon & <v_{i}, \label{onpath}\\
\forall i \neq j,  \;v_{i}(i) + \varepsilon & <v_{i}(j). \label{jdev}
\end{align}
As usual, we interpret each $v(i)$ as giving each player $j\neq i$ a `reward' for punishing $i$.
Extend each such finite sequence to a periodic sequence by defining
$\widetilde a^i_t:=\widetilde a^i_{t \; \text{mod} \; T}$ for $t\geq T$. For each $i$, the punishment profile for
player $i$ is a mixed action profile $\alpha^i \in \mathcal{A}^{\textrm{MI}}$ such that 
\begin{equation}
{\alpha^{i}}\in\arg\min_{\alpha\in\mathcal{A}^{\textrm{MI}}}\max_{a_{i}}g_{i}(a_{i},\alpha_{-i}).
\end{equation}
Choose an $N\in \mathbb N$ such that for all $i$,
\begin{equation}
\max_{ a\in A}g_{i}( a)+N\, g_{i}({\alpha^{i}})
<(N+1)(v_{i}(i)-\varepsilon), \ \ \ \forall i,
\label{choose-N}
\end{equation}
which is possible because of $g_i(\alpha^i)\leq \underline v_i^{\mathrm{MI}}$ and (\ref{idev}).

Choose a $\hat{\delta}$ high enough that for all $\delta>\hat{\delta}$, the implication of Theorem \ref{DG} holds. Then choose $\underline\delta \geq \hat{\delta}$ so that all of the following hold for each $i$ and each $\delta \geq \underline\delta$:
\begin{align}
\tag{Rewards}\label{Rewards}
&&\forall \, k\in \mathbb Z_+,~~\left\|\frac{1-\delta}{1-\delta^T}\sum_{t=0}^{T-1} \delta^{t}g(\widetilde a^i_{(t+k)})-v(i)\right \|<\varepsilon,\\
\tag{IC-I} \label{IC-I}
&&(1-\delta)\max_{ a\in A}g_{i}(a)+\delta[(1-\delta^{N})g_i(\alpha^i)+\delta^{N}v_{i}(i)]<(1-\delta)\,\min_{a\in A}g_{i}( a)+\delta(v_{i}-\varepsilon),\\
&&\tag{IC-II(i)}\label{IC-II(i)}
\underline v^{\mathrm{MI}}_i<(1-\delta^{N})g_i(\alpha^i)+\delta^{N} (v_i(i)-\varepsilon),\\
&&\label{IC-II(j)}\tag{IC-II(j)}
\forall t\leq N, \forall j\not =i,~~ (1-\delta)\max{g_i(a)} +\delta[(1-\delta^N)g_i(\alpha^i)+\delta^Nv_i(i))] \\
\nonumber\ \     && <(1-\delta^t)g_i(\alpha^j)+\delta^t(v_i(j)-\varepsilon),\\
&&\label{IC-III(i)}\tag{IC-III(i)}
(1-\delta)\max_{a\in A}g_{i}( a)+\delta(1-\delta^N)g_i(\alpha^i)<(1-\delta^{N+1})(v_i(i)-\varepsilon),\\
&&\label{IC-III(j)}\tag{IC-III(j)}
\forall j\not = i, ~~(1-\delta)\max{g_i(a)}+\delta[(1-\delta^N)g_i(\alpha^i)+\delta^Nv_i(i)]<v_i(j)-\varepsilon.
\end{align}
For \eqref{Rewards} such a choice is possible by (\ref{Rewardssetup}),
and for \eqref{IC-I}
(resp.~\eqref{IC-II(i)}, \eqref{IC-II(j)}, \eqref{IC-III(i)}, \eqref{IC-III(j)})
because its
limit as $\delta\uparrow 1$ reduces to $v_i(i)<v_i-\varepsilon$
(resp.~$\underline v_i^\mathrm{MI}< v_i(i)-\varepsilon$,
$v_i(i)<v_i(j)-\varepsilon$,
(\ref{choose-N}),
$v_i(i)<v_i(j)-\varepsilon$),
which holds by (\ref{onpath})
(resp.~(\ref{idev}), (\ref{jdev}), the choice of $N$, (\ref{jdev})).

\emph{Strategies.}

For any $\delta>\underline \delta$, we define a strategy profile that 
is a Blackwell equilibrium above $\delta$.
Play is based on the following phases:\\
 Phase I: Play $a( v,\varepsilon,\delta)$, a sequence of pure action profiles satisfying (\ref{SAactions_v}).\\
Phase II(i): Play ${\alpha^{i}}$ for $N$ rounds.\\
Phase III(i): Play $\widetilde a^i$, starting at $\widetilde a^i_0$.

We construct a simple strategy profile \`a la Abreu (1988): We start in Phase I; unilateral deviations by a player $j$ from any phase lead to Phase II(j) followed by Phase III(j).

As we used Theorem \ref{thm:DG2017} to generate Phase I, the payoffs of the specified strategies evaluated at discount $\delta$ are $v$.
It remains to show that for any $\delta^\prime\geq \delta$, the strategies form an SPNE; \textit{i.e.}, that the strategies are a Blackwell SPNE above $\delta$.

For any $\delta^\prime\in (0,1)$, let $v_{i}^{t}(\delta^\prime)$ and $v_{i}^{t}(j)(\delta^\prime)$ denote the $\delta^\prime$-discounted
continuation payoff of the path in Phase I and Phase III(j) respectively, after $t-1$ rounds of the corresponding phase (not of the entire game) have elapsed. Note that, differently from FM and standard perfect-monitoring folk theorems, we do not ask a player to (myopically) best respond during her own punishment phase, as that would potentially not leave the others willing to mix.

From Theorem \ref{thm:DG2017} and the Patience Lemma, we have 
\begin{equation}
\forall i, \forall t, \forall \delta^\prime\geq \delta, ~~v_{i}^{t}(\delta^\prime)\geq v_{i}-\varepsilon. \label{ineq1}
\end{equation}
From (\ref{Rewards}) we have 
\begin{equation}
\forall i, \forall j, \forall t, \forall \delta^\prime\geq \delta, ~~|v^t_{i}(j)(\delta^\prime)-v_i(j)|<\varepsilon \label{ineq3}.
\end{equation}
From the fact that $g_i(\widetilde a^i_t)\leq g_i(\widetilde a^i_{t+1})$ for $t\in \{0,1,...T-2\}$, we have 
\begin{align}
\forall i, \forall \delta^\prime\geq \delta,~~ v^1_{i}(i)(\delta^\prime)&\leq v_{i}(i) ~~~\text{and}\label{ineq2}\\
\forall i, \forall t, \forall \delta^\prime\geq \delta, ~~v^1_{i}(i)(\delta^\prime)&\leq v^t_{i}(i) (\delta^\prime)\label{ineq4}.
\end{align}

\emph{Checking subgame perfection}.\\
 \emph{Step 1.} Player $i$ cannot profit by deviating from Phase I if for any
 $ t \in \mathbb{N}$,
\begin{equation*}
(1-\delta^\prime)\,{\displaystyle \max_{ a\in A}g_{i}( a)+\delta^\prime[(1-\delta^{\prime N})g_i(\alpha^i)+\delta^{\prime N}v^1_{i}(i)(\delta^\prime)]\leq(1-\delta^\prime)\,\min_{ a\in A}g_{i}(a)+\delta^\prime v_{i}^{t}(\delta^\prime).}
\end{equation*}
Using (\ref{ineq1}), (\ref{ineq2}) and $g_i(\alpha^i)\leq \underline v_i^{\mathrm{MI}}$ we need only show
\[
(1-\delta^\prime)\max_{a\in A}g_{i}(a)+\delta^\prime[(1-\delta^{\prime N})\underline v^{\mathrm{MI}}_{i}+\delta^{\prime N}v_{i}(i)]\leq(1-\delta^\prime)\,\min_{ a\in A}g_{i}(a)+\delta^\prime(v_{i}-\varepsilon),
\]
which is identical to (\ref{IC-I}), which applies as $\delta^\prime\geq \delta\geq \underline \delta$.\\
\emph{Step 2.} Player $i$ cannot profit by deviating from Phase II(i) if for any $t = 1,2,\dotsc, N$, 
\[(1-\delta^\prime)\underline v^{\mathrm{MI}}_i+\delta^\prime[(1-\delta^{\prime N})g_i(\alpha^i)+\delta^{\prime N} v^1_i(i)(\delta^\prime)]<(1-\delta^{\prime t})g_i(\alpha^i)+\delta^{\prime t} v^1_i(i)(\delta^\prime).
\]
Using (\ref{ineq3}) and $g_i(\alpha^i)<v_i^{\mathrm{MI}}\leq v_i(i)-\varepsilon$ from (\ref{idev}), we can get the sufficient condition
\[(1-\delta^\prime)\underline v^{\mathrm{MI}}_i+\delta^\prime[(1-\delta^{\prime N})g_i(\alpha^i)+\delta^{\prime N} (v_i(i)-\varepsilon)]<(1-\delta^{\prime N})g_i(\alpha^i)+\delta^{\prime N} (v_i(i)-\varepsilon),
\]
which reduces to (\ref{IC-II(i)}), which applies as $\delta^\prime\geq \delta\geq \underline \delta$.
\\
 \emph{Step 3.} Player $i$ cannot profit by deviating from Phase III(i) if for any $t \in \mathbb{N}$, 
\[
(1-\delta^\prime)\,\max_{a\in A}g_{i}( a)+\delta^\prime[(1-\delta^{\prime N})\, g_{i}({\alpha^{i}})+\delta^{\prime N}v_{i}^{1}(i)(\delta^\prime)]\leq v_{i}^{t}(i)(\delta^\prime).
\]
Given (\ref{ineq4}), this inequality holds if
\[
(1-\delta^\prime)\,\max_{a\in A}g_{i}(a)+\delta^\prime(1-\delta^{\prime N})\, g_{i}({\alpha^{i}})\leq (1-\delta^{\prime N+1}) v_{i}^{t}(i)(\delta^\prime),
\]
so that we can now use (\ref{ineq3}) to get the sufficient condition
\[
(1-\delta^\prime)\,\max_{a\in A}g_{i}(a)+\delta^\prime(1-\delta^{\prime N})\, g_{i}({\alpha^{i}})\leq (1-\delta^{\prime N+1}) (v_{i}(i)-\varepsilon),
\]
which is satisfied due to (\ref{IC-III(i)})
 given that $\delta^\prime\geq \delta\geq \underline \delta$.\\
 \emph{Step 4.} Player $i$ does not deviate (observably) from Phase
II(j) if for all remaining punishment rounds $t\leq N$
\begin{equation*}
 (1-\delta^\prime)\max{g_i(a)}+\delta^\prime[(1-\delta^{\prime N})g_i(\alpha^i)+\delta^{\prime N}v^1_i(i)(\delta^\prime)]<(1-\delta^{\prime t})( g_i(\alpha^j)+\delta^{\prime t}v^1_i(j)(\delta^\prime).
\end{equation*}
It suffices to use (\ref{ineq2}) and (\ref{ineq3}) to obtain the sufficient condition
\begin{equation*}
 (1-\delta^\prime)\max{g_i(a)}+\delta^\prime[(1-\delta^{\prime N})g_i(\alpha^i)+\delta^{\prime N}v_i(i)(\delta^\prime)]<(1-\delta^{\prime t})( g_i(\alpha^j)+\delta^{\prime t}(v_i(j)-\varepsilon),
\end{equation*}
which is (\ref{IC-II(j)}), which applies as $\delta^\prime\geq \delta\geq \underline \delta$.\\
 \emph{Step 5.} Player $i$ cannot profit by mixing differently in Phase II(j). Recall our definition of $\mathcal A^{\mathrm{MI}}$; since ${\alpha^{j}}\in\mathcal{A}^{\mathrm{MI}}$, mixing only occurs between myopically indifferent actions according to ${\alpha^{j}}$; as future play does not vary over $i$'s actions on $\supp(\alpha_{i}^{j})$, he does not have a strict incentive to deviate. \\
\emph{Step 6.} Player $i$ cannot profit by deviating from Phase III(j) if for any $t \in \mathbb{N}$,
\begin{equation*}
(1-\delta^\prime)\max{g_i(a)}+\delta^\prime[(1-\delta^{\prime N})g_i(\alpha^i)+\delta^{\prime N}v^1_i(i)(\delta^\prime)]<v_i^t(j)(\delta^\prime).
\end{equation*}
so that using (\ref{ineq2}) and (\ref{ineq3})  we can use (\ref{IC-III(j)}) as a sufficient condition.

Therefore for any $\delta^\prime \geq \delta$, the specified strategies form an SPNE, and hence they are a Blackwell SPNE above $\delta$.

If a PRD is available, the equilibrium strategies can be modified as follows:\\
Phase I: At each round, play the correlated action $p\in\Delta A$
such that $v=\sum_{a\in A}\, p( a)g(a)$.\\
Phase II(i): Play ${\alpha^{i}}$ for $N$ rounds.\\
Phase III(i): Play $p^{i}\in\Delta A$ such that ${v(i)}=\sum_{a\in A}\, p^{i}( a)g( a)$ at each round of the phase.

It is easy to see that the resulting strategies constitute a Blackwell SPNE
if $\delta$ satisfies the sufficient condition $\delta\geq\underline{\delta}$ in our PRD-free construction above.
\end{proof}

\section*{Appendix B: Proofs for Section \ref{sec:IM} (Imperfect Monitoring)}

\subsection*{Proofs of Lemmata \ref{lem:ballSS} and \ref{lem:RobustEq}}

Following FLM 1994, we proceed by first showing a local version of the self-generation property we seek.

\begin{defn}
A set $W\subset\mathbb{R}^{n}$ is said to be \textbf{locally strong self-generating} if for any $v\in W$ we can find an open set $\mathcal{O}_{v}$, an $\eta_v>0$ and a $\delta_{v}<1$ such that
\[
v\in\mathcal{O}_{v}\cap W\subset \mathcal B_{\eta_v}(W;\delta)  \;  \forall \delta \geq \delta_v.
\]
\end{defn}

We will show that a closed ball $B(c,r) \subset \mathbb{R}^n$ is locally strong self-generating if it lies in the interior of the set $F^{\mathrm{MI}, \pi}$.

To this end it is useful to introduce the notion of MI-score. 
Following Matsushima (1989) and FL 1994, for any non-zero direction $\lambda$ we can find a point $v^{*}(\lambda)$ that lies on the highest hyperplane in direction $\lambda$ subject to the point itself being generated by a current action in $\mathcal{A^{\mathrm{MI}}}$, and continuation payoffs that lie below the said hyperplane.\footnote{Matushima (1989) proposed an algorithm to characterize the upper boundary of the equilibrium payoff set, when first-order conditions suffice for a maximum; FL 1994 extended this to all directions to characterize the entire set, restricting attention to finite action spaces to enable sufficient conditions to be imposed explicitly.} Let $H^-$ denote the lower half-space function, \textit{i.e.}, $H^-(\lambda,k):=\{z\in \mathbb R^n: \lambda\cdot z\leq k\}$. 
Let $\mathcal{B}(\cdot, \alpha; \delta)$ denote the usual APS operator for a fixed current action profile $\alpha$. That is, $v\in \mathcal{B}(W, \alpha; \delta)$ if there is a $w:Y \rightarrow W$  such that $v$ is the payoff of the Nash equilibrium $\alpha$ of the normal-form game with payoffs $(1-\delta)g(a)+\delta \mathbb E[w|a]$. 

\begin{defn}
The \textbf{MI-score} in direction $\lambda$ is
\begin{equation}
    k^{{\mathrm{MI}}} (\lambda) := \;
\sup_{v \in \mathbb{R}^n } 
\bigg \{ \lambda \cdot v  \bigg|
v \in \underset{\alpha \in \mathcal{A^\mathrm{MI}}}
\bigcup \mathcal{B} (H^{-}(\lambda, \lambda \cdot v), \alpha ;\delta) \bigg \}.
\label{eq:MI-score}
\end{equation}
\end{defn}

The region bounded by the MI-score in each direction serves as an upper bound on the equilibrium payoff set.
The difference between the usual score and the MI-score in \eqref{eq:MI-score} is that we restrict the current action $\alpha$ to have the myopic indifference property. 
The computation of the score is, loosely speaking, more flexible than the computation of self-generation because for any direction it allows us to generate payoffs using continuations in the lower half-space rather than in a smaller self-generating set.  
Note also that incentives aren't strict in calculating the score, since we use the standard APS operator $\mathcal{B}$ rather than our strong version $\mathcal{B}_{\eta}$ for some $\eta >0$.
While our equilibrium construction requires strict incentives, this is tackled perturbing it to create a point with almost the maximal score but with strict incentives.

Online Appendix OA shows how IFR implies that 
$F^{\mathrm{MI}, \pi}=\{v\in \mathbb R^n \big| \lambda \cdot v \leq k^{\mathrm{MI}}(\lambda)~~ \forall \lambda \neq 0 \}$.
So to prove Lemma \ref{lem:ballSS} it suffices to show that a closed ball $B(c,r)$ is locally strong self-generating if the MI-score in any direction $\lambda \neq 0$ exceeds the value $\lambda \cdot v$ for any $ v\in B(c,r)$.

Any $v \in B(c,r)$ falls in one of two cases.

{\it Case 1}: $v\in \interior(B(c,r))$.

We can find $\mu>0$ such that
$B(v,\mu)\subset \interior(B(c,r))$.
Fixing a NE $\alpha^{*}$ of the stage game $G$, we pick $\delta_{v}$ large enough so that the implied continuation payoffs at each $v^\prime\in B(v,\mu)$ (which are chosen to be constant in both the signal and the output of the randomization device) lie in $ \interior(B(c,r))$.

{\it Case 2}: $v\in\partial B(c,r)$.

The unit normal vector at $v$ pointing away from $B(c,r)$ is 
$\lambda := (v-c)/ \| v-c \|$. If there is an NE of $G$ that lies above this hyperplane we follow the arguments as in Case 1.
If not, pick a point $v^{*}$ that gives the maximal $\mathrm{MI}$-score in direction $\lambda$ and find the associated mixed action $\alpha^*\in \mathcal A^{\textrm {MI}}$ and $x^*:Y\rightarrow \mathbb R^n$, the associated normalized continuation payoff, \footnote{For details, see Mailath and Samuelson (2006). For all $\delta\in(0,1)$, $(\alpha^{\ast},\frac{1-\delta}{\delta}x^{*}+v^*)$ generates $v^{\ast}$ on $H^{-}(\lambda,\lambda\cdot v^{\ast})$ at $\delta$.}
 which satisfies that for all $y$, $\lambda\cdot x^*(y)\leq 0$.

For each $i$, order the actions of $i$ as $A_i = \{a_{i,1}, \ldots ,a_{i,K} \}$.
Define $v'_i \in \mathbb{R}^{|A_i|}$ by introducing a `penalty' of $1$ for any action that is not in the support of the action that generates $v^{*}$:
$$v^\prime_{i,k} = \begin{cases}
v_i ~~~&\text{if }a_{i,k}\in \supp(\alpha^{\ast}_i)\\
v_i-1~~~&\text{otherwise}.
\end{cases}
$$
By IFR, the matrix $\Pi_i (\alpha)$ whose rows are transposes of the column vectors $ \pi(y| \alpha_{-i}, a_i) $, one for each $a_i \in A_i$, has full row rank; therefore the following linear equation has a solution $x^\prime_i\in \mathbb R^{|Y|}$:
\begin{equation}\begin{bmatrix} 
    v^\prime_{i,1} -  g_i(a_{i,1},\alpha^*_{-i}) \\
    \vdots  \\
    v^\prime_{i,K} -  g_i(a_{i,K},\alpha^*_{-i})   \end{bmatrix}
 = \Pi_i (\alpha^{*}) 
 \begin{bmatrix} 
    x^\prime_i(y_1) \\
    \vdots  \\
    x^\prime_{i}(y_{|Y|})  \end{bmatrix}.
\end{equation}

Therefore, we've converted our demand for strictness in payoffs into a signal-specific reward function. Having done this for each $i$, define the function $x^\prime:Y \rightarrow \mathbb{R}^n$ by combining the $x^\prime_i$, \textit{i.e.}, $x^\prime (y)=\left( x^\prime_1(y), \ldots,  x_n^\prime(y) \right)$. 
Take $\beta,\gamma\in(0,1)$ small enough that for all $y$, we have $0>\lambda \cdot (\beta\gamma x^\prime(y)+\beta(1-\gamma)(v-v^*)+(1-\beta\gamma)x^*(y))$; this is possible from $\lambda \cdot x^*(y)\leq 0$ and $\lambda\cdot(v-v^*)<0$. Let $x^{\prime\prime}(y):=\beta\gamma x^\prime(y)+\beta(1-\gamma)(v-v^*)+(1-\beta\gamma)x^*(y)$, so that $\lambda \cdot x^{\prime\prime}(y)<0$.

Let $v^{\prime\prime}:=\beta v+(1-\beta)v^*$.
Since $v^*$ maximises the $\mathrm{MI}$-score and $\lambda \cdot v<k^{\textrm{MI}}(\lambda)$,
 we have  $\lambda\cdot v < \lambda\cdot v^*$,
  and therefore
$v^{\prime\prime}$ satisfies $\lambda\cdot v < \lambda\cdot v''$.

Therefore, the following hold for all $i$:
\begin{align}
 v''_i & = g_i(a_i,\alpha^{*}_{-i})
           +\mathbb{E}(x''_i \vert (a_i,\alpha_{-i}^*)),
                 & \textrm{ if } a_i \in \supp (\alpha_i^{*}), \\
 v''_i & \geq g_i(a_i,\alpha_{-i}^*)
     +\mathbb{E}(x''_i \vert (a_i,\alpha_{-i}^*))+ \beta\gamma, & \textrm{ otherwise.}
\end{align}
We define $\eta_v=\beta\gamma$, the degree of slack we've introduced.

If $\lambda$ is not a negative coordinate direction, $\alpha^*$ can be chosen to be pure; otherwise, perform the following `garbling' to ensure \eqref{eq:wdist}:
\begin{itemize}
\item
Publicly draw $\nu$ from the uniform distribution on $[0,1]^n$.%
\footnote{Combined with Lemma \ref{lemma:reboot}, an $(n+1)$-dimensional public randomization device suffices in each round, one dimension each to garble the continuation payoffs and an extra one for the reboot decision.}
\item
Let $\underline{x}_j$ be the lowest value of the of $x_j^{\prime\prime}$ on $Y$, and let $\overline{x}_j$ be the highest.
\item
Let $\widetilde{x}(y,\nu) = (\widetilde{x}_j(y_j,\nu_j))_{j\in I}$ be given by
$$
\widetilde{x}_j(y_j,\nu_j) =
\begin{cases}
\overline{x}_j, &\text{if }\nu_j \leq \frac{x_j^{\prime\prime}(y_j)-\underline{x}_j}{\overline{x}_j-\underline{x}_j},\\
\underline{x}_j, &\text{otherwise. }
\end{cases}
$$
\end{itemize}
Note that $\widetilde{x}_j$ is a garbling of $x_j^{\prime\prime}$ that preserves the expectation, and that
$$\mathrm{Pr}(\widetilde{x}_j = \overline{x}_j \mid a_j) 
= \sum_{y_j\in Y_j} \pi_j(y_j\mid a_j) \frac{x_j^{\prime\prime}(y_j)-\underline{x}_j}{\overline{x}_j-\underline{x}_j},$$
which is constant on $\supp(\alpha_j)$.\footnote{The degenerate case of $\overline{x}_j = \underline{x}_j$ is trivial.}
This needs to be done only for negative-coordinate directions.
For all other directions, set
$\widetilde{x}_j(y_j,\nu_j) = x_j^{\prime\prime}(y_j)$ for all $\nu$. Thus
\begin{align}
 v''_i & =  g_i(a_i,\alpha^{*}_{-i})+\mathbb{E}_{y,\nu}(\widetilde x_i \vert (a_i,\alpha_{-i}^* )),
                 & \textrm{ if } a_i \in \supp (\alpha_i^*), \\
 v''_i & \geq  g_i(a_i,\alpha^{*}_{-i})
     +\mathbb{E}_{y,\nu}(\widetilde x_i \vert (a_i,\alpha_{-i}^*))+ \eta_v, & \textrm{ otherwise.}
\end{align}

Since $ \lambda \cdot v < \lambda \cdot v'' $, add $ v - v''$ to both sides to  translate the old normalised continuation function to a new one $x$:
$$v =g(\alpha^{*})+\mathbb{E}_{y,\nu}(x \vert\alpha^*), \mbox{ where } x \; := \widetilde x_j + v -v''.$$
The continuation payoff point
$w(y,\nu; \delta) = (w_j(y_j,\nu_j;\ \delta))_{j\in I}$ is given by
\begin{equation}
w_j(y_j,\nu_j;\delta) = v_j +
\frac{1-\delta}{\delta}{x}_j(y_j,\nu_j),\label{eqn:continuations}
\end{equation}
which satisfies
${\delta}^2 \left\Vert w(y,\nu ;\delta)-c\right\Vert ^2
    =  (1-\delta)^2 \left\Vert x \right\Vert ^2 + 2 \delta (1-\delta) x\cdot (v-c) + \delta^2 r^2$.
Now, using the fact that $\lambda=\frac{v-c}{\left\Vert v-c \right\Vert}$, we have  $x \cdot (v-c) ={\left\Vert v-c \right\Vert} \lambda \cdot (\widetilde x + v -v'')
= {\left\Vert v-c \right\Vert}\lambda \cdot \widetilde x+ {\left\Vert v-c \right\Vert}(\lambda \cdot v -  \lambda \cdot v'') <0$, by $\lambda \cdot \widetilde x< 0$ and $ \lambda \cdot v < \lambda \cdot v'' $. Thus there is a $\delta_{y,\nu}$ such that for $\delta\geq\delta_{y,\nu}$, the continuation payoff $w(y, \nu; \delta)$ lies in the interior of $B(c,r)$. Take $\delta_{v}$ to be $\max_{y,\nu}\delta_{y,\nu}$, which is less than $1$ as each $\delta_{y,\nu}<1$. At each $\delta>\delta_v$, we have
\begin{align}
v_i & = (1-\delta) g(\alpha_i,\alpha^{*}_{-i})+\delta\mathbb{E}_{y,\nu}(w_j(\cdot;\delta)|(\alpha_i,\alpha^*_{-i})),
                 & \textrm{ if } a_i \in \supp (\alpha_i^*), \label{eqn:exact}\\
v_i & \geq (1-\delta) g(\alpha_i,\alpha^{*}_{-i})+\delta \mathbb{E}_{y,\nu}(w_j(\cdot;\delta)|(\alpha_i,\alpha^*_{-i})) + (1-\delta)\eta_v, & \textrm{ otherwise.}\label{eqn:slack}
\end{align}

Thus
\[
 v \in \mathcal B_{\eta_v} (B(c,r); \delta) \;  \forall \delta \geq \delta_v.
\]

Since translating the continuation payoff function leaves incentives unaffected, there exists $\gamma_v>0$ such that all points in $B(c,r)\cap B(v,\gamma_v)$ can be generated by using continuations in $B(c,r)$ at $\delta_v$, preserving (\ref{eqn:exact}) and (\ref{eqn:slack}). Moreover, at $\delta>\delta_v$ translated continuation values are in $B(c,r)$, as they are convex combinations of translated original points and translated continuations at $\delta_v$, both of which are in $B(c,r)$. Thus, for all $\delta>\delta_v$, we have $B(v,\gamma_v) \cap B(c,r) \subset \mathcal B_{\eta_v} (B(c,r);\delta)$.

Combining Cases 1 and 2 we see that $B(c,r)$ is locally strong self-generating. Now we can leverage our local
self-generation property into a global one.

Let $Z\subset \interior F^{\mathrm{MI}, \pi}$ be a compact locally strong self-generating set.
 Since $Z$ is locally strong self-generating, for any $z\in Z$ there exists $\gamma_{z}>0$ and $\eta_z>0$ such that $B(z,\gamma_{z})\cap Z$ can be $\eta_z$-strongly generated by $Z$ at all $\delta>\delta_z$. Then $\{(\interior B(z,\gamma_{z}))\cap Z\}_{z\in Z}$ forms an open cover of $Z$. Since $Z$ is compact, extract a finite subcover $\left\{ B(z_{1},\gamma_{z_{1}}),\dotsc,B(z_{L},\gamma_{z_{L}})\right\}$.
Now let $\delta_{Z}:=\max\{\delta_{z_{1}},\dotsc,\delta_{z_{L}}\}$, which is strictly below $1$ as a maximum of finitely many reals strictly less than $1$; and $\eta^*:=\min \{\eta_{z_{1}},...,\eta_{z_{L}}\}$, which is positive as the minimum of finitely many positive reals.

For all $\delta>\delta_Z$, since  $B(z_{\ell},\gamma_{z_{\ell}})\cap Z \subset \mathcal B_{\eta^*}  (B(c,r);\delta)$, and
$\cup_{\ell=1}^{L}\left(B(z_{\ell},\gamma_{z_{\ell}})\cap Z\right)=Z$, the set $Z$ is $\eta^*$-strongly self-generating.
\qed

\subsection*{Proof of Theorem \ref{thm:NoPRD}}

The following doubly uniform version of robustness is important in this setting.
\begin{defn}
A pair $(W,\hat \Delta)\in F \times 2^{(0,1)}$ is a \textbf{doubly robust region} if for any $v^\prime\in W$ and any $\delta\in \hat \Delta$,
we can find a strategy profile that (i) is a PPE for any discount
factor in $\hat \Delta$,
and (ii) delivers the payoff $v^\prime$ at $\delta$.
\end{defn}

\begin{lemma}
\label{lem:double-robust-reg} Given a closed ball $B\subset \interior F^* $, there exists $\underline{\delta}\in(0,1)$ such that for any $\delta^{\prime}\in(\underline{\delta},1)$ there is an open interval
$(\delta_l,\delta_h)\ni \delta^\prime$ such that $(B,(\delta_l,\delta_h))$ is a doubly robust region.
\end{lemma}
\begin{proof}
First, we notice that Lemma \ref{lem:ballSS} only used the PRD to get local self-generation in the negative coordinate direction. In our case, we are free to use either a pure or Nash action as in each negative coordinate direction $-e_i$ as we only need to attain a score of $\min\{\underline{v}_i^{\mathrm{pure}} , \underline{v}_i^{\textrm{NE}}\}$. Thus garbling is unnecessary, and strictness only has to be introduced when $\alpha^*$ is pure, which by assumption implies it satisfies IFR. Thus, a modified version of Lemma \ref{lem:ballSS} holds with respect to closed balls in $\interior F^*$, using a PRD-free version of the strong APS operator.

Thus, there is a $\underline\delta\in (0,1)$ and an $\eta>0$ such that for all $\delta\in (\underline \delta,1)$ we have $B\subset \mathcal B_{\eta}(B;\delta)$. Take any $\delta^\prime>\underline \delta$; we will show there is an open interval satisfying the lemma. 

Let $M:=\max_{i,a}|g_i(a)|$. Take $\delta_l,\delta_h$ such that 
\begin{align}
\underline \delta < \delta_l< \delta^\prime <\delta_h \label{deltainequalities}\\
\delta_l>\delta_h-\frac{\eta(1-\delta_h)^2}{4M} \label{etatodelta}
\end{align}

Given any $\delta\in (\delta_l,\delta_h)$ and any $v\in B$, using  that $\underline\delta<\delta$ from (\ref{deltainequalities}), construct a PPE using Lemma \ref{lem:ballSS} and $\eta$-strong self-generation at $\delta$ with payoff $v$. Then, take any $\hat \delta\in (\delta_l,\delta_h)$. Take any public history and any action profile $\alpha$ played after that history. From the definition of strong self-generation, for any player $i$ and action $a_i$ one of the following is true:
\begin{itemize}
\item The distribution of continuation play does not vary between $\alpha$ and $a_i,\alpha_{-i}$. Then, the fact $\alpha$ is a Nash equilibrium implies $g_i(\alpha)\geq g_i(a_i,\alpha_{-i})$  and hence $a_i$ is never a better response than $\alpha_i$ regardless of the discount.
\item Playing $a_i$ entails a loss of at least $\eta(1-\delta)$ at discount $\delta$. Notice the average discounted payoff function is differentiable in $\delta$ and for any outcome $h^\infty$ we have the bound $|\frac{d}{d\delta_0}U_i(h^\infty,\delta_0)|\leq \frac{2M}{1-\delta_0}$. Thus, between any two $\delta_0,\delta_1 \in (\delta_l,\delta_h)$ each strategy can vary in payoffs by at most $\frac{2M}{1-\max\{\delta_0,\delta_1\}}$. Over the whole interval, then, which from (\ref{etatodelta}) has length of at most $\frac{\eta(1-\delta_h)^2}{4M}$, this variation is bounded by $\eta(1-\delta_h)/2$. Thus, at $\hat \delta$, deviating to $a_i$ entails a gain of less than $2\eta(1-\delta_h)/2-\eta(1-\delta)<0$; so the deviation is unprofitable.
\end{itemize}
Thus, the constructed strategies are a PPE at any $\hat \delta$ in $(\delta_l,\delta_h)$, as requested.
\end{proof}

Now, we begin the proof of Theorem \ref{thm:NoPRD} in earnest. First, we construct statistical tests using $(a^*,Y^*)$. These tests will be used to mimic a PRD and decide whether the game should or should not be reset. A success occurs when the signal is $Y^{*}$. Recall that any unilateral deviation from $a^{\ast}$ strictly reduces the probability of a signal in the set $Y^{\ast}$. Let $\mathcal{T}(T^{*},k^{*})$ denote a test that is passed if and only if there are at least $k^{*}$ successes in $T^{*}$ Bernoulli trials, each with success probability $q^{*}=\pi(Y^{*}|a^{*})$. The pass probability of the test is the probability of passing it if $a^{*}$ is played for $T^{*}$ rounds.

Play is divided into three phases --- \textit{Select}, \textit{Normal}, and \textit{Test}. Play begins in the \textit{Select} phase.
Intuitively, the \textit{Select} is used to randomize whether the players next find themselves in a reward or punishment \textit{Normal} phase. \textit{Normal} phases are where players collect most of their payoffs. Each round $n$ in a Normal phase is linked via incentives to the $n$th round in future Normal phases, until a reset is triggered. Following a \textit{Normal} phase, play moves to the \textit{Test} phase; if the test is failed, then the subsequent \textit{Select} phase randomizes over punishment and reward once again; otherwise, the next \textit{Normal} phase will continue where play left off.

Given a target payoff $v$, choose the following quantities in order, culminating in the choice of a cutoff discount factor $\delta_{\min}$.
\begin{enumerate}

\item[Step 1.] Choose $c^{+},c^{-} \in \interior F^*$ such that $c_{i}^{+}>v_{i}>c_{i}^{-}$ for every $i$
and the following holds:
\begin{equation}
v=q^{*}c^{+}+(1-q^{*})c^{-}.\label{eq:v-split}
\end{equation}

\item[Step 2.] Choose $r>0$  such that $c_{i}^{+}>v_{i}+5r;\;v_{i}>c_{i}^{-}+5r\;\forall i$; and
$$B(c^{-},5r),B(c^{+},5r) \subset \interior \{v\in \co u(A) \mid  \forall i\in I, v_i\geq \underline{v}_i^{\mathrm{pure}} \wedge \underline{v}_i^{\textrm{NE}}
\}.$$
As we will use balls of radius $r$ rather than $5r$, the distance along each coordinate between any point in the balls we will use and the point $v$ is at least $4r$.

\item[Step 3.] Choose a $\delta^{\prime}$ high enough to satisfy the conditions
of Lemma \ref{lem:double-robust-reg} for both $B(c^{-},r)$ and $B(c^{+},r)$. That is, for some $(\delta_l^-,\delta_h^-)\ni \delta^\prime$ and $(\delta_l^+,\delta_h^+)\ni \delta^\prime$ we have that $(B(c^{-},r),(\delta_l^-,\delta_h^-))$ and $(B(c^{+},r),(\delta_l^+,\delta_h^+))$ are both doubly robust regions. Take $(\underline{\delta},\overline{\delta}):=(\delta_l^-,\delta_h^-)\cap(\delta_l^+,\delta_h^+)$.

\item[Step 4.] Find two binomial tests, both based on $(a^{*},Y^{*})$, with pass
probabilities $1-p^{+}$ and $p^{-}$, such that
\begin{equation}
1-p^{+},1-p^{-}\in(\underline{\delta},\overline{\delta}),\textrm{ and }\label{eq:1-p}
\end{equation}
\begin{equation}
\left\| v-\frac{\frac{q^{*}c^{+}}{1-(1-p^{+})\widehat{\delta}}+\frac{(1-q^{*})c^{-}}{1-(1-p^{-})\widehat{\delta}}}{\frac{q^{*}}{1-(1-p^{+})\widehat{\delta}}+\frac{1-q^{*}}{1-(1-p^{-})\widehat{\delta}}} \right\|
< \frac{r}{2}, \ \ \;\forall\widehat{\delta}\in\left(\max\{\frac{\underline{\delta}}{1-p^{+}},\frac{\underline{\delta}}{1-p^{-}}\},1\right).\label{cond:small-lambda}
\end{equation}
Such tests can always be found. Since the binomial distribution approximates the normal distribution (which has a continuous CDF), we can use independent Bernoulli trials with success probability $q^{*}=\pi(Y^{*}\mid a^{*})$ to design binomial tests of suitable lengths $T^{+}$ and $T^{-}$ having pass probabilities arbitrarily close to $\frac{\underline \delta+\bar\delta}{2}$ and $1-\frac{\underline \delta+\bar\delta}{2}$ respectively. \
If the approximations are close enough, (\ref{eq:1-p}) holds and so does  inequality (\ref{cond:small-lambda}), because equation (\ref{eq:v-split}) implies that the LHS of inequality (\ref{cond:small-lambda}) goes to zero as $p^+$ and $p^-$ get closer.

\item[Step 5.] As soon as the probability of the current test being passed hits zero or unity, we cannot give incentives to players to play the test action
unless it is a Nash equilibrium. To ensure incentives throughout a
test, we show how to modify it by truncating it appropriately. Given
any test of length $T^{0}$, we can switch to playing a fixed Nash
equilibrium $\alpha^{NE}$ of the stage game as soon as the test is
conclusive (failed or passed for sure) and until $T^{0}$ rounds
are up. Let $\mathcal{T}^{-}$ and $\mathcal{T}^{+}$ denote, respectively,
the Nash-truncated versions of the minus and plus Binomial tests above
extended to a common length $T:=\max\{T^{-},T^{+}\}$; if $T^{-}<T^{+}=T$,
where Nash play is substituted following the conclusion of a test.
If the continuation play after any test phase depends only on the
test outcome (pass vs fail), playing $\alpha^{NE}$ is incentive compatible
for any discount factor. Since the Nash equilibrium is played only
when the test is conclusive, the Nash-truncated tests inherit the
pass probabilities of the original tests.

\item[Step 6.] Tests not only impact subsequent play but also distort payoffs and incentives. We therefore must ensure they account for only a small proportion of all periods. This is both so that we can deliver the target payoff, and so that normal phases provide incentives to play the test action
during test phases. Denoting by $\rho$ the minimum change in the passing probability of a test when a player deviates from her test action,\footnote{There are two tests, finitely many states in each test, finite action
spaces and $n<\infty$ players, so this $\rho$ is well-defined and, due to the assumed properties of the
test action profile $a^*$, is strictly positive.} choose $N$
large enough that
\begin{equation}
N>6\frac{\sqrt{n}M(T+1)}{r\rho},\label{eq:N_1}
\end{equation}
where $|u_{i}(a)|\leq M$ for all $a\in A,i\in I$.

\item[Step 7.] We now ensure that we can incentivize players to take
the test action. We have
already (in Step 6) found a $N$ large enough to make most periods `normal' - what remains is
to find a bound on $\delta$ above which normal periods provide sufficient incentives during the
test phases.
Pick a $\widehat{\delta}$ such that for all $\delta\geq\widehat{\delta}$ and $i\in I$,
\begin{align*}
2MT+\frac{\delta^{T}}{1-\delta}\rho\left[\frac{(c^{-}_i+3r/2)(1-\delta^{N})-(1-\delta^{(N+T+1)})(v_i-3r)}{1-(1-p^{-})\delta^{(N+T+1)}}\right]<0,
\end{align*}
and
\begin{align*}
2MT-\frac{\delta^{T}}{1-\delta}\rho\left[\frac{(c^{+}-3r/2)(1-\delta^{N})-(1-\delta^{(N+T+1)})(v_i+3r)}{1-(1-p^{+})\delta^{(N+T+1)}}\right]<0.
\end{align*}

In the limit as $\delta \rightarrow 1$, these become
\begin{align}
\frac{2MT}{\rho}p^-+(T+1)(v_i-3r)<N\left (v_i-c_i^--\frac{9r}{2}\right)\\
\frac{2MT}{\rho}p^++(T+1)(v_i+3r)<N\left (c_i^+-v_i-\frac{9r}{2}\right).
\end{align}

Given that Step 2 guarantees that $c_i^+-5r>v_i>c_i^-+5r$, both of these are implied by (\ref{eq:N_1}). Thus there is a $\hat \delta$ as we require.

\item[Step 8.] 
We require that the payoff distortion due to tests is small so that we can achieve our target payoff. We have
(\ref{eq:N_1}), but it does not take into account discounting. Thus, we find
a $\delta_{\text{dist}}\in(0,1)$ such that for all $\delta>\delta_{\text{dist}}$
we have  
\begin{equation}
\sum_{t=T+1}^{N+T}\delta^{t}r\rho/2>\sum_{t=0}^{T}\delta^{t}3\sqrt{n}M,\label{eq:discN_1}
\end{equation}
which is possible as in the limit as $\delta\rightarrow 1$ this inequality reduces to the undiscounted
version in (\ref{eq:N_1}), which is strict.

\item[Step 9.] We want the equilibrium payoff to not vary much with the
discount. To do this we find $\delta_{s}$ such that for
all $i\in I$ and all $\delta>\delta_{s}$ we have
\[
\left[1-\delta\frac{1-\delta^{N}}{1-\delta^{(N+T+1)}}\right]|v_{i}|<r.
\]
As $\delta$ increases to $1$, the limiting condition is $(T+1)|v_i|<(N+T+1)r$, which is implied by Step 6; therefore,
such a cutoff $\delta_s$ can be found.

\item[Step 10.] Finally, define
\begin{equation}
\delta_{\min}:=\max\left\{ \left(\frac{\underline{\delta}}{1-\max\{p_{-},p_{+}\}}\right)^{\frac{1}{N+T+1}},\delta_{\text{dist}},\widehat{\delta},\delta_{s}\right\} .
\end{equation}
\end{enumerate}
Given a $\delta>\delta_{\min}$ , proceed as follows.

Define $\delta^{+}:=\delta^{N+T+1}(1-p^{+})$ and $\delta^{-}=\delta^{N+T+1}(1-p^{-})$.
By construction,
\begin{equation}
\delta^{+},\delta^{-}\in(\underline{\delta},\overline{\delta})
\end{equation}

First, we want to adjust the starting points of the self-generation
algorithm in both balls to account for (slightly) different reset
probabilities, but without taking into account the test rounds. Define
\begin{equation}
\lambda\,:=v-\left[\frac{\frac{q^{*}(c^{+})}{1-\delta^{+}}+\frac{(1-q^{*})(c^{-})}{1-\delta^{-}}}{\frac{q^{*}}{1-\delta^{+}}+\frac{1-q^{*}}{1-\delta^{-}}}\right].
\end{equation}
The expression in square brackets has the following interpretation.
Suppose that we toss a coin with a probability $q^{\ast}$ of coming
up heads; if we toss heads we play a sequence of actions that gives us $c^{+}$
while if it comes up tails we play the sequence of actions that gives
us $c^{-}$. However, as we reset back to the beginning with different probabilities,
we do not get $q^*c^++(1-q^*)c^-$. Instead, we end up getting a combination
of $c^{+}$ and $c^{-}$with weights in the ratio $q^{\ast}/(1-\delta^{+})$
to $(1-q^{\ast})/(1-\delta^{-})$. The quantity $\lambda$ can be
interpreted as the adjustment to the starting point needed to get
$v$; this follows as the above can be rearranged to write $v$ as
a convex combination of $c^{+}+\lambda$ and $c^{-}+\lambda$:
\[
v=\frac{\frac{q^{*}(c^{+}+\lambda)}{1-\delta^{+}}+\frac{(1-q^{*})(c^{-}+\lambda)}{1-\delta^{-}}}{\frac{q^{*}}{1-\delta^{+}}+\frac{1-q^{*}}{1-\delta^{-}}}.
\]
By (\ref{cond:small-lambda}) we have that $ \left\| \lambda \right\| <r/2$.

We now further adjust the starting points in view of the non-normal
phases. We compute the payoff $v^{+}$ such that one round of the average \textit{Select} phase
payoff $z$ (detailed below),
then $N$ rounds of $v^{+}$, and finally $T$ rounds of the test
$\mathcal{T}^{+}$ (whose average $\delta$-discounted payoffs we denote $x_{\mathcal T^+}$) gives the same expected utility as getting $c^{+}+\lambda$ for $(1+N+T)$ rounds, \textit{i.e.},
\begin{equation}
(1-\delta^{N+T+1})(c^{+}+\lambda)=(1-\delta)z+\delta(1-\delta^{N})v^{+}+\delta^{N+1}(1-\delta^{T})\,x_{\mathcal T^+}.\label{eq:definev+}
\end{equation}
Subtracting $v^+$ from both sides and using the magnitude operator and the triangle inequality, we have
\begin{equation}
(1-\delta^{N+T+1})||(c^{+}+\lambda-v^+)||\leq (1-\delta)||(z-v^+)||+\delta^{N+1}(1-\delta^{T})\,||(x_{\mathcal T^+}-v^+)||.\label{eq:definev+}
\end{equation}
and replacing payoff differences by $2M$, the largest possible value, we have
\begin{equation}
(1-\delta^{N+T+1})||(c^{+}+\lambda-v^+)||\leq (1-\delta)2\sqrt{n}M+\delta^{N+1}(1-\delta^{T})\,2\sqrt{n}M,
\end{equation}
from which we deduce
\begin{equation}
||(c^{+}+\lambda-v^+)||< \frac{(1-\delta^{T+1})}{(1-\delta^{N+T+1})}2\sqrt{n}M.
\end{equation}
Finally, we use Step 8 to deduce
\begin{equation}
||(c^{+}+\lambda-v^+)||< \frac{\delta^{T+1}-\delta^{N+T+1}}{(1-\delta^{N+T+1})}\frac{r}{2}<\frac{r}{2}.
\end{equation}

Combining this with $ \|  \lambda \| <r/2$, the triangle inequality implies
that $ \|  c^{+}-v^{+} \| <r$; thus $v^{+}\in B(c^{+},r)$. Define $v^-$ similarly, which guarantees $v^{-}\in B(c^{-},r)$.
By the above, Step 3, and Lemma \ref{lem:double-robust-reg},
there exists a strategy profile $\sigma^{-}$ that
is a PPE for all $\delta^{\prime\prime}\in(\underline{\delta},\overline{\delta})$,
while delivering $v^{-}$ at $\delta^{-}$, and similarly a $\sigma^+$ that is a PPE
in that same set of discounts delivering $v^+$ at $\delta^+$.

\textit{Strategy Profile.} Now we describe the strategy $\sigma(\delta,v)$ as an automaton over 4 types of phases\textemdash $Select$ (1 round), $Test$ ($T$ rounds), $Norm+$($N$ rounds), and $Norm-$($N$
rounds).
\begin{itemize}
\item Start in $Select$.
\item $Select$:
\begin{itemize}
\item In the first round of play or if the test immediately before triggered
a reset, play $a^{*}$ once; if the selection
test succeeds (which it does with probability $q^{*}$), move to $Norm+$,
else move to $Norm-$.
\item If the preceding test did not trigger a reset, play $\alpha^{NE}$
once, and then move to a Normal phase of the same type
as the one leading up to the test.
\end{itemize}
\item $Norm-$: If this is the first normal phase, or the last test concluded
in favor of a reset, start playing $\sigma^-$ using an empty history in
each of the following $N$ rounds; if not, play according to $\sigma^-$
where each of the $N$ cycles left off. That is, round $t$ uses
the action profile $\sigma^-(h_{t}^{*})$ with the
linked history $h_{t}^{*}$ defined by $h_{t}^{*}:=(y_{t-n(N+T+1)},y_{t-(n-1)(N+T+1)},...y_{t-(N+T+1)})$
where $n$ is the number of tests since the last reset. 
At the end of this phase move to $Test$.
\item $Norm+$: As above, with plus instead of minus.
\item $Test$: Play the Nash-truncated test $\mathcal{T}^{-}$, lasting
$T$ rounds, if the immediately preceding normal phase was $Norm-$;
else play $\mathcal{T}^{+}$, also lasting $T$ rounds. Then move
to $Select$.
\end{itemize}

As the next result shows, the payoff under $\sigma(\delta,v)$ is close to $v$, for every high enough discount factor. For its proof, see the online appendix.

\begin{lemma}[Bounding payoffs] \label{lem:bounding-payoffs}
For any $\delta^{*}\geq\delta$,  $|U_i(\sigma(\delta,v),\delta^*)-v_i|<3r$ for all $i\in I$.
\end{lemma}

\subsubsection*{Incentives}
We now show that the strategy $\sigma(\delta,v)$ is a PPE at any
$\delta^{\ast}\geq\delta_{\min}$.

At a history $h$ that leads to the start of a \textit{Norm-} phase, the future payoffs due
to play according to $\sigma^-$ until a reset are
\begin{equation}
\sum_{j=1}^{N}\delta^{\ast j-1}\sum_{\tau=0}^{\infty}[(1-p^{-})\delta^{\ast(N+T+1)}]^{\tau}g_i(\sigma^-(h^{*}))
\end{equation}
and reflect the $N$ `separate' cycles of play according to the self-generation algorithm. Play during the test and select phases
gives at most $M$ per period, so that the payoff due to test and reset phases until the next reset are bounded above by
\begin{equation}
\sum_{\tau=0}^{\infty}[(1-p^{-})\delta^{\ast(N+T+1)}]^{\tau}[M(T+1)].
\end{equation}
Finally, a reset returns continuation payoffs to $U_i(\sigma(\delta,v),\delta^*)$, so that total future payoffs from play after resets are bounded above
by 
\begin{equation}
\delta^{\ast(N+T+1)}\sum_{\tau=0}^{\infty}[\delta^{\ast(N+T+1)}(1-p^{-})]^{\tau}p^{-}\frac{U_i(\sigma(\delta,v),\delta^*)}{1-\delta^{\ast}}
\end{equation}

Thus, we can bound payoffs at the outset of a $Norm-$ phase, following a history $h$, by

\begin{align}
\nonumber
&&\frac{U_i(\sigma(\delta,v)(h),\delta^*)}{1-\delta^{\ast}} \\
&&\leq \left[\sum_{j=1}^{N}\delta^{\ast j-1}\sum_{\tau=0}^{\infty}[(1-p^{-})\delta^{\ast(N+T+1)}]^{\tau}g_i(\sigma^-(h^{*}))
 \label{equ:inc:1} +\sum_{\tau=0}^{\infty}[(1-p^{-})\delta^{\ast(N+T+1)}]^{\tau}[M(T+1)]\right]\\
&&\ \ \ +  \delta^{\ast(N+T+1)}\sum_{\tau=0}^{\infty}[\delta^{\ast(N+T+1)}(1-p^{-})]^{\tau}p^{-}\frac{U_i(\sigma(\delta,v),\delta^*)}{1-\delta^{\ast}}.
\nonumber
\end{align}

Consider a deviation by some player during a $Test$ phase. By the
fact that $a^{*}$ signal-maximizes $Y^{*}$, deviation from $a^{*}$
therefore decreases the probability of passing the test. Suppose
the test is following a $Norm-$ phase; failing the test therefore
decreases the probability of a reset. Take the smallest such change,
$\rho$. An upper bound on the within-test-phase benefits of a deviation
is $2MT$. So, a deviation is not profitable if
\begin{equation}
2MT+\delta^{\ast T}\rho\left[\frac{U_i(\sigma(\delta,v)(h),\delta^*)-U_i(\sigma(\delta,v),\delta^*)}{1-\delta^{*}}\right] \leq 0,
\label{equ:inc:2}
\end{equation}
where $h$ is the relevant history.

We can substitute the expression in \eqref{equ:inc:1} - an upper bound to $U_i(\sigma(\delta,v)(h),\delta^*)$ - to bound the LHS of \eqref{equ:inc:2} above by

\begin{align*}
  2MT  +\delta^{\ast T}\rho
 &
 \bigg[
 \sum_{j=1}^{N}\delta^{\ast j-1}
 \sum_{\tau=0}^{\infty} \{ (1-p^{-})\delta^{\ast(N+T+1)} \}^{\tau}g_i(\sigma^-(h^{*}))\\
 & +\sum_{\tau=0}^{\infty}
 \{(1-p^{-})\delta^{\ast(N+T+1)} \}^{\tau} M(T+1) \\
 & + \big\{ \delta^{\ast(N+T+1)}\sum_{\tau=0}^{\infty}
 \{ \delta^{\ast(N+T+1)}(1-p^{-}) \}^{\tau}p^{-}-1 \big\} \frac{U_i(\sigma(\delta,v),\delta^*)}{1-\delta^{\ast}}
 \bigg],
\end{align*}
which is in turn bounded above (using Step 8) by
\begin{align*}
2MT+\delta^{\ast T}\rho \bigg[ \sum_{j=1}^{N}\delta^{\ast j-1}\sum_{\tau=0}^{\infty}[(1-p^{-})\delta^{\ast(N+T+1)}]^{\tau}[g_i(\sigma^-(h^{*}))+r/2]\\ -\frac{1-\delta^{\ast(N+T+1)}}{1-(1-p^{-})\delta^{\ast(N+T+1)}}\frac{U_i(\sigma(\delta,v),\delta^*)}{1-\delta^{\ast}} \bigg].
\end{align*}
Now we can apply the Patience Lemma to get a new bound:
\begin{align*}
2MT+\delta^{\ast T}\rho\left[\sum_{j=1}^{N}\delta^{\ast j-1}\frac{c^{-}_i+3r/2}{1-(1-p^{-})\delta^{\ast(N+T+1)}}-\frac{1-\delta^{\ast(N+T+1)}}{1-(1-p^{-})\delta^{\ast(N+T+1)}}\frac{U_i(\sigma(\delta,v),\delta^*)}{1-\delta^{\ast}}\right].
\end{align*}
Using our bounds on $U_i(\sigma(\delta,v),\delta^*)$ from Lemma (\ref{lem:bounding-payoffs}), we can get a further upper bound,
\begin{align*}
2MT+\frac{\delta^{\ast T}}{1-\delta^{\ast}}\rho\left[\frac{(c^{-}_i+3r/2)(1-\delta^{*N})-(1-\delta^{\ast(N+T+1)})(v_i-3r)}{1-(1-p^{-})\delta^{\ast(N+T+1)}}\right],
\end{align*}
which we know from Step 7 is negative. Therefore, there is no incentive
to deviate during a $Test$ phase following $Norm-$. The same argument
applies mutatis mutandis to a $Test$ phase following $Norm+$, as
well as the $Select$ phase.\\

Incentives during the $Norm+$ and $Norm-$ Phases: Since the strategy
in the normal phase involves playing a PPE strategy with rebooting,
it follows that there is no incentive for any player to unilaterally
deviate from the prescribed actions in any round.

\subsubsection*{Equilibrium payoffs}

Given the construction of $v^{+}$ and $v^{-}$, the payoff vector
of this equilibrium at discount $\delta$ is $U(\sigma(\delta,v),\delta)=v$.
This completes our proof. \hfill \qed

\newpage
\section*{Online Appendix OA: Additional Proofs for Section \ref{sec:IM}}

The goal of this online appendix is two-fold. First, we show that the intersection over all non-zero directions of the lower half-spaces defined by the corresponding MI-scores equals the set $F^{\textrm{MI}, \pi}$. Second, we prove Lemma 7. 

The intersection of the maximal lower half-spaces is
\[
H^{*} := \underset{\lambda \neq 0}{\cap} H^{-} (\lambda, k^{\mathrm{MI}}).
\]
The MI-score differs from the usual score in that only action profiles with myopic indifference can be used to achieve it. Note also that incentives aren't strict in calculating the score, since we use the standard APS operator $\mathcal{B}$ rather than our strict version $\mathcal{B}_{\eta}$ for some $\eta >0$.
While our equilibrium construction requires strict incentives, this is tackled perturbing it to create a point with almost the maximal score but with strict incentives.

We consider the following mutually exclusive directions: (i) non-coordinate directions, \textit{i.e.}, at least two coordinates are non-zero; (ii) positive coordinate directions, \textit{i.e.}, exactly one coordinate is $+1$, while the rest are zero;
(iii) negative coordinate directions, \textit{i.e.}, exactly one coordinate is $-1$ while the rest are zero.

\begin{lemma} \label{lem:H*isFX}
For any direction $\lambda$, the $\mathrm{MI}$-score is attained by an action profile in $\mathcal{A^{\mathrm{MI}}}$; furthermore, a pure action profile achieves this score in directions other than negative coordinate ones.  Additionally,  $F^{\mathrm{MI}, \pi}=\{v\in \mathbb R^n \big| \lambda \cdot v \leq k^{\mathrm{MI}}(\lambda)~~ \forall \lambda~[\|\lambda\|=1]\}$.
\end{lemma}

The proof of the lemma is standard. Intuitively, IFR allows us to generate incentives in all but negative-coordinate directions. 
These require us to use mixed actions in general, and indeed, the Blackwell restriction implies that any such action must also be in $\mathcal{A}^{\mathrm{MI}}$; this is just the MI-minmax defined earlier, rather than the standard minmax seen in the literature.

\begin{proof}
 We divide the proof into the above cases.

 \noindent Case 1: $\lambda$ is a non-coordinate direction. One of the highest points in $F$ in direction $\lambda$ must be generated by a pure action profile, which is in $\mathcal{A^{\mathrm{MI}}}$, and using IFR we can satisfy the IC conditions with equality using continuation payoffs in the lower halfspace --- this is the same argument as in Lemma 5.1 in FL 1994.

 \noindent Case 2: $\lambda$ is a positive coordinate direction ${e}_i$. The score maximization reduces to $\max v_i$ subject to the IC constraints. We pick the pure action profile that maximises the payoff of $i$ in $F$. Thanks to IFR we can design continuation payoffs so that incentives hold with equality.

 \noindent Case 3: $\lambda$ is a negative coordinate direction $-{e}_i$.
Only here do might we need to achieve the $\mathrm{MI}$-score using a mixed action which is in $\mathcal{A}^\mathrm{MI}$. In this case it follows from our definition of the program $P_i^{\mathrm{MI},\pi_i}$ and FL 1994 that the score is the negative of the minmax $\underline v_i^{\textrm{MI}, {\pi}_i}$.
\end{proof}

\subsection*{Proof of Lemma \ref{lem:bounding-payoffs}.}
We first bound $U_i(\sigma(\delta,v),\delta^*)$ from above. Notice that terms are
divided based on whether they come from the plus or the minus ball.
Within each ball we can group terms into $1+N+T$
cycles, and finally sum over the $N$ normal phase cycles, replacing
the test phase cycles by the largest possible payoff. This gives
\begin{align*}
\frac{U_i(\sigma(\delta,v),\delta^*)}{1-\delta^{\ast}}
\leq & q^{*} \sum_{j=1}^{N}\delta^{\ast j}\sum_{\tau=0}^{\infty}[(1-p^{+})\delta^{\ast(N+T+1)}]^{\tau}g_i(\sigma^+(h^{*})) \\
+ &  q^{*} \sum_{\tau=0}^{\infty}[(1-p^{+})\delta^{\ast(N+T+1)}]^{\tau}[M(T+1)] \\
+ & (1-q^{*}) \sum_{j=1}^{N}\delta^{\ast j}\sum_{\tau=0}^{\infty}[(1-p^{-})\delta^{\ast(N+T+1)}]^{\tau}g_i(\sigma^-(h^{*})) \\
+ & (1-q^{*}) \sum_{\tau=0}^{\infty}[(1-p^{-})\delta^{\ast(N+T+1)}]^{\tau}[M(T+1)]\\
+ & q^{*}\delta^{\ast(N+T+1)}\sum_{\tau=0}^{\infty}[\delta^{\ast(N+T+1)}(1-p^{+})]^{\tau}p^{+}\frac{U_i(\sigma(\delta,v),\delta^*)}{1-\delta^{\ast}}\\
+ & (1-q^{*})\delta^{\ast(N+T+1)}\sum_{\tau=0}^{\infty}[\delta^{\ast(N+T+1)}(1-p^{-})]^{\tau}p^{-}\frac{U_i(\sigma(\delta,v),\delta^*)}{1-\delta^{\ast}}
\end{align*}
which, by Step 6, grants
\begin{align*}
\frac{U_i(\sigma(\delta,v),\delta^*)}{1-\delta^{\ast}}< & q^{*}\sum_{j=1}^{N}\delta^{\ast j}\sum_{\tau=0}^{\infty}[(1-p^{+})\delta^{\ast(N+T+1)}]^{\tau}(g_i(\sigma^+(h^{*}))+r/2)\\
+ & (1-q^{*})\sum_{j=1}^{N}\delta^{\ast j}\sum_{\tau=0}^{\infty}[(1-p^{-})\delta^{\ast(N+T+1)}]^{\tau}(g_i(\sigma^-(h^{*}))+r/2)\\
+ & q^{*}\delta^{\ast(N+T+1)}\sum_{\tau=0}^{\infty}[\delta^{\ast(N+T+1)}(1-p^{+})]^{\tau}p_{+}\frac{U_i(\sigma(\delta,v),\delta^*)}{1-\delta^{\ast}}\\
+ & (1-q^{*})\delta^{\ast(N+T+1)}\sum_{\tau=0}^{\infty}[\delta^{\ast(N+T+1)}(1-p^{-})]^{\tau}p^{-}\frac{U_i(\sigma(\delta,v),\delta^*)}{1-\delta^{\ast}}.
\end{align*}
Since these actions are generated by the self-generation algorithm,
which keeps all continuation payoffs in a ball of radius $r$ around
$c^{-}$ or $c^{+}$, the Patience Lemma gives a bound on each cycle:
\begin{align*}
\frac{U_i(\sigma(\delta,v),\delta^*)}{1-\delta^{\ast}}< & q^{*}\sum_{j=1}^{N}\delta^{\ast j}\frac{(c^{+}_i+r+r/2)}{1-(1-p_{+})\delta^{\ast(N+T+1)}}+(1-q^{*})\sum_{j=1}^{N}\delta^{\ast j}\frac{(c^{-}_i+r+r/2)}{1-(1-p^{-})\delta^{\ast(N+T+1)}}\\
+ & q^{*}\delta^{\ast(N+T+1)}\sum_{\tau=0}^{\infty}[\delta^{\ast(N+T+1)}(1-p_{+})]^{\tau}p_{+}\frac{U_i(\sigma(\delta,v),\delta^*)}{1-\delta^{\ast}}\\
+ & (1-q^{*})\delta^{\ast(N+T+1)}\sum_{\tau=0}^{\infty}[\delta^{\ast(N+T+1)}(1-p^{-})]^{\tau}p^{-}\frac{U_i(\sigma(\delta,v),\delta^*)}{1-\delta^{\ast}}.
\end{align*}
Collecting all terms with $v(\delta^{*})$ on the left, we have
\begin{align*}
\frac{U_i(\sigma(\delta,v),\delta^*)}{1-\delta^{\ast}}\left[1-\frac{q^{\ast}p^{+}\delta^{\ast(N+T+1)}}{1-(1-p^{+})\delta^{\ast(N+T+1)}}-\frac{q^{\ast}p^{-}\delta^{\ast(N+T+1)}}{1-(1-p^{-})\delta^{\ast(N+T+1)}}\right]\\
<\sum_{j=1}^{N}\delta^{\ast j}\frac{q^{\ast}(c^{+}_i+r+r/2)}{1-(1-p^{+})\delta^{\ast(N+T+1)}}+\frac{(1-q^{\ast})(c^{-}_i+r+r/2)}{1-(1-p^{-})\delta^{\ast(N+T+1)}}.
\end{align*}
This in turn leads to
\begin{align*}
\frac{U_i(\sigma(\delta,v),\delta^*)}{1-\delta^{\ast}}\left[\frac{q^{*}(1-\delta^{\ast(N+T+1)})}{1-(1-p^{+})\delta^{\ast(N+T+1})}+\frac{(1-q^{*})(1-\delta^{\ast(N+T+1)})}{1-(1-p^{-})\delta^{\ast(N+T+1})}\right]\\
<\delta^{\ast}\frac{1-\delta^{\ast N}}{1-\delta^{*}}\left[\frac{q^{\ast}(c^{+}_i+r+r/2)}{1-(1-p^{+})\delta^{\ast(N+T+1)}}+\frac{(1-q^{\ast})(c^{-}_i+r+r/2)}{1-(1-p^{-})\delta^{\ast(N+T+1)}}\right],
\end{align*}
or simply
\begin{align*}
U_i(\sigma(\delta,v),\delta^*)<\delta^{\ast}\frac{1-\delta^{\ast N}}{1-\delta^{\ast(N+T+1)}}\left[3r/2+\frac{\frac{q^{\ast}c_i^{+}}{1-(1-p^{+})\delta^{\ast(N+T+1)}}+\frac{(1-q^{\ast})c^{-}_i}{1-(1-p^{-})\delta^{\ast(N+T+1)}}}{\frac{q^{*}}{1-(1-p^{+})\delta^{\ast(N+T+1})}+\frac{(1-q^{*})}{1-(1-p^{-})\delta^{\ast(N+T+1})}}\right].
\end{align*}
This and Step 4 grants
\[
U_i(\sigma(\delta,v),\delta^*)<\delta^{\ast}\frac{1-\delta^{\ast N}}{1-\delta^{\ast(N+T+1)}}[2r+v_i],
\]
or,
\[
U_i(\sigma(\delta,v),\delta^*)-v_i<\delta^{\ast}\frac{1-\delta^{\ast N}}{1-\delta^{\ast(N+T+1)}}2r-[1-\delta^{\ast}\frac{1-\delta^{\ast N}}{1-\delta^{\ast(N+T+1)}}]v_i,
\]
so that Step 9 gives
\[
U_i(\sigma(\delta,v),\delta^*)-v_i<3r.
\]
Reasoning similarly with the lower bound implies
\[
U_i(\sigma(\delta,v),\delta^*)-v_i>-3r.
\]
The required bound follows. \hfill \qed
\bigskip

\section*{Online Appendix OB: Proofs for Section \ref{sec:Ext}}


\subsection*{Proof of Theorem \ref{LBPPM}}

We start with a Tauberian result, which is used to show that if the average of the first $T$ terms of a sequence of reals converge as $T\to \infty$, then the normalized discounted sum of the whole sequence converges to the same limit as $\delta\uparrow 1$. The following classic result, due to Frobenius, is our starting point.

\begin{lemma}[Frobenius] \label{frob} For any sequence of reals $(x_{t})_{t=0}^{\infty}$ satisfying
\[
\lim_{T\rightarrow\infty}\frac{1}{T+1}\sum_{t=0}^{T}\,\sum_{k=0}^{t}x_{t}=x^{\ast}\in\mathbb{R},
\]
we have $\lim_{\delta\uparrow1}\sum_{t=0}^{\infty}\delta^{t}x_{t}=x^{\ast}$.
\end{lemma}

\begin{corr}\label{summeq}
If for a sequence of reals $(x_{t})_{t=0}^{\infty}$ we have $\lim_{T\rightarrow\infty}\frac{1}{T+1}\sum_{t=0}^{T}x_{t}=x^{\ast}\in \mathbb R$, then $\lim_{\delta\uparrow1}(1-\delta)\sum_{t=0}^{\infty}\delta^{t}x_{t}=x^\ast$.
\end{corr}
\begin{proof}
Defining $x_{-1}=0$, we have
\begin{equation}
x^{\ast}=\lim_{T\rightarrow\infty}\frac{1}{T+1}\sum_{t=0}^{T}x_{t}=\lim_{T\rightarrow\infty}\frac{1}{T+1}\sum_{t=0}^{T}\sum_{k=0}^{t}(x_{k}-x_{k-1})
\end{equation}
as $\sum_{k=0}^{t}(x_{k}-x_{k-1})=x_{t}$. Using Lemma \ref{frob},
we obtain
\[
x^{\ast}=\lim_{\delta\uparrow1}\sum_{t=0}^{\infty}\delta^{t}(x_{t}-x_{t-1})=\lim_{\delta\uparrow1}\left[\sum_{t=0}^{\infty}(\delta^{t}-\delta^{t+1})x_{t}\right]=\lim_{\delta\uparrow1}(1-\delta)\sum_{t=0}^{\infty}\delta^{t}x_{t}.
\hfill     \qedhere
\]
\end{proof}

In light of this, if we can construct a Blackwell SPNE whose on-path play yields an undiscounted average payoff of $v$, then $v$ is a limit Blackwell payoff. This leads to the proof of the folk theorem.

Fix $v\in F^{\mathrm{MI}}$. We wish to construct
a strategy profile $\sigma$ such that for any given $\varepsilon>0$
there is some $\ud<1$ above which $\sigma$ is a BE with discounted
payoff within $\varepsilon$ of $ v$.

In view of Corollary \ref{summeq}, feasibility reduces
to obtaining $\bs v$ as the limit of means of a sequence of pure
action payoffs. While this is easy to do, some care is needed because
in order for a sequence of actions to be an equilibrium path, we also
need to make sure that the (discounted) continuation payoffs are individually
rational; in fact our insistence on Blackwell equilibria means that
we should keep continuation payoffs of each $i\in I$ above the corresponding
$\mathrm{MI}$-minmax value $\underline v_{i}^{\mathrm{MI}}$, not just the usual
minmax $\underline v_{i}$.

The crux is that if the target payoff $v$ is the discounted sum of pure-action payoffs, \textit{i.e.}, of points in $C=g(A)$, all continuation payoffs are not necessarily individually rational even if  $v$ is.
To circumvent this, we represent $v$ as the discounted sum of individually rational payoffs.  To this end, construct a full-dimensional set $D=\{d(1),...,d(K)\}$, where $K \in \mathbb{N}$, such that
(i) $v\in \interior \{\co(D)\}\subset \co(C)$;
(ii) each $ d(k)\in D$ is a rational convex combination of points
in $C$; and
(iii) there exists $\gamma>0$ for which $d_{i}(k)>\underline v_{i}^{\mathrm{MI}}+3\gamma$ for all $i$ and all $k=1,2,\dotsc,K$.
Since $v \in \co({D})$, there exists a weight vector (non-negative components adding up to unity) $\lambda = (\lambda(k))_{k=1}^K$  such that $ v=\sum_{k=1}^{K}\lambda(k) d(k)$.
Let  $(\lambda^m)_m$ be a sequence of weight vectors with rational components such that $ \lambda^m \rightarrow \lambda$ as $m \rightarrow \infty$. Let $v^{m}:=\sum_{k=1}^{K}\lambda^{m}(k)d(k)$.
Since each $d(k)$ is a rational convex combination of points in $C$,
a finite sequence (`subcycle $k$') from $C$ averages to $d(k)$. Without loss of generality, these $k$ subcycles have the same length.\footnote{If not, find the least common multiple $L$ of the lengths of the subcycles, and repeat each of the subcycles the appropriate number of times to create $K$ new subcycles each of length $L$.}  Similarly we write $v^{m}$ as a finite sequence (`cycle $m$') of points in $D$. Concatenate the cycles for $m=0,1,2,...$; then
replace each occurance of $d(k)$ in each cycle by subcycle $k$
to create the sequence $(x_{t}:t\geq0)$ of payoff profiles in
$C$, called the payoff path; the corresponding sequence of actions
is the action path. Since there are finitely many distinct subcycles,
choosing $\widetilde{\delta}<1$ high enough ensures that for $\delta\geq\widetilde{\delta}$
the following two conditions hold---(i) the $\delta$-discounted sum
of any subcycle differs from the simple mean by at most $\gamma$;
(ii) $(1-\delta^{L})M<\gamma$, where $L$ is the maximum length of
a subcycle and all individual payoffs of the stage game are in $[-M,M]$.
Property (i) implies that any $\delta$-discounted continuation payoff
of the path from the start of any subcycle is at least $\underline v_{i}^{\mathrm{MI}}+2\gamma$;
properties (i) and (ii) together imply that the continuation payoff
of the path from any time (even when it is not the start of a subcycle)
is at least $\underline v_{i}^{\mathrm{MI}}+\gamma$.

The means of the cycles are $v^{m}$ and since $\| v^{m}- v\| \rightarrow0$,
we have $\frac{1}{T+1}\sum_{t=0}^{T} x_{t}\rightarrow v$; Corollary
\ref{summeq} then implies that for any  some $\widehat{\delta}\in(\widetilde{\delta},1)$
we have
\begin{equation}
\left\|  v-(1-\delta)\sum_{t=0}^{\infty}\delta^{t}x_{t}\right\|
\leq\varepsilon \; \forall \, \delta\geq\widehat{\delta}.\label{eq:approx_bound}
\end{equation}

The rest of the construction, as well as the proof that the resulting strategies
constitute a Blackwell Equilibrium, follows the same path as the proof of Theorem \ref{thm:perfect}.

\end{document}